\documentclass[a4paper]{article}

\usepackage{microtype}
\usepackage{amsmath,amssymb,amsthm}
\usepackage{fullpage}
\usepackage{graphicx}
\usepackage{subfigure}
\usepackage{xspace}
\usepackage{hyperref}

\newtheorem{theorem}{Theorem}
\newtheorem{proposition}{Proposition}
\newtheorem{lemma}{Lemma}

\newcommand{\A}{\mathcal{A}}
\newcommand{\B}{\mathcal{B}}

\newcommand{\R}{\mathbb{R}}

\newcommand{\tA}{\text{\tiny A}}
\newcommand{\tB}{\text{\tiny B}}
\newcommand{\tL}{\text{\tiny L}}
\newcommand{\tR}{\text{\tiny R}}
\newcommand{\tLR}{\text{\tiny LR}}
\newcommand{\tRL}{\text{\tiny RL}}
\newcommand{\tSR}{\text{\tiny SR}}

\newcommand{\dl}{d_{\tL}}
\newcommand{\dr}{d_{\tR}}
\newcommand{\dlr}{d_{\tLR}}
\newcommand{\drl}{d_{\tRL}}
\newcommand{\deltalr}{\delta_{\tLR}}
\newcommand{\deltarl}{\delta_{\tRL}}
\newcommand{\blr}{\beta_{\tLR}}
\newcommand{\brl}{\beta_{\tRL}}

\newcommand{\mul}{\mu_{\tL}}
\newcommand{\mur}{\mu_{\tR}}
\newcommand{\as}{\alpha^{\ast}}

\newcommand{\sigmas}{\sigma^{\ast}}

\newcommand{\sigmaa}{\sigma_{\tA}}
\newcommand{\deltaa}{\delta_{\tA}}
\newcommand{\SQUARE}{\Gamma}
\newcommand{\SQUAREAB}{\square}
\newcommand{\RECTII}{\Xi}
\newcommand{\lL}{\textsc{l}}
\newcommand{\lR}{\textsc{r}}
\newcommand{\lC}{\textsc{c}}
\newcommand{\lB}{\textsc{b}}
\newcommand{\lA}{\textsc{a}}

\newcommand{\Bc}{B^{\circ}}
\newcommand{\Bb}{\bar{B}}
\newcommand{\LER}{\Lambda}

\newcommand{\SL}{S^{\tL}}
\newcommand{\SR}{S^{\tR}}
\newcommand{\FL}{F^{\tL}}
\newcommand{\FR}{F^{\tR}}

\newcommand{\ds}{d^{\ast}}

\newcommand{\pth}[2][\!]{#1\left({#2}\right)}

\newcommand{\relphantom}[1]{\mathrel{\phantom{#1}}}

\newenvironment{denseitems}{\list{$\bullet$}%
  {\labelwidth3em\itemsep0pt\parsep0pt\topsep0.6ex}}{\endlist}

\newcommand{\pccc}{\ensuremath{\mathit{CCC}}}
\newcommand{\pcsc}{\ensuremath{\mathit{CSC}}}
\newcommand{\plrl}{\ensuremath{\mathit{LRL}}}
\newcommand{\prlr}{\ensuremath{\mathit{RLR}}}

\newcommand{\prsr}{\ensuremath{\mathit{RSR}}}
\newcommand{\plsl}{\ensuremath{\mathit{LSL}}}
\newcommand{\prsl}{\ensuremath{\mathit{RSL}}}
\newcommand{\plsr}{\ensuremath{\mathit{LSR}}}

\newcommand{\AD}{A^{\Delta}}
\newcommand{\BD}{B^{\Delta}}
\newcommand{\CD}{C^{\Delta}}

\DeclareMathOperator{\dub}{dub}
\renewcommand{\l}{\ell}
\newcommand{\Th}{\vartheta}
\DeclareMathOperator{\llsl}{\textsc{lsl}}
\DeclareMathOperator{\lrsr}{\textsc{rsr}}
\DeclareMathOperator{\llsr}{\textsc{lsr}}
\DeclareMathOperator{\lrsl}{\textsc{rsl}}
\DeclareMathOperator{\lrslp}{\textsc{rsl}^{\ast}}
\DeclareMathOperator{\llrl}{\textsc{lrl}}
\DeclareMathOperator{\lrlr}{\textsc{rlr}}

\newcommand{\arc}[1]{\ddddot{#1}}

\let\geq\geqslant
\let\leq\leqslant

\makeatletter
\ifx\showkeyslabelformat\@undefined\else
\renewcommand{\showkeyslabelformat}[1]{\normalfont\tiny\ttfamily#1}
\fi
\partopsep\z@ \textfloatsep 10pt plus 1pt minus 4pt
\def\section{\@startsection {section}{1}{\z@}{-3.5ex plus -1ex minus
-.2ex}{2.3ex plus .2ex}{\large\bf}}
\def\subsection{\@startsection{subsection}{2}{\z@}{-3.25ex plus -1ex
minus -.2ex}{1.5ex plus .2ex}{\normalsize\bf}}
\def\@fnsymbol#1{\ensuremath{\ifcase#1\or *\or 1\or 2\or
    3\or 4\or 5\or 6\or 7 \or 8\ or 9 \or 10\or 11 \else\@ctrerr\fi}}
\makeatother

\begin{document}

\title{The Cost of Bounded Curvature%
  \thanks{This research was supported in part by Mid-career Researcher
    Program through NRF grant funded by the~MEST
    (No.~R01-2008-000-11607-0) and in part by the NRF grant
    2011-0030044 (SRC-GAIA) funded by the government of Korea.}}

\author{Hyo-Sil Kim%
  \thanks{Department of Computer Science, KAIST, Daejeon, Korea.
    \{hyosil,otfried\}@tclab.kaist.ac.kr.}
  \and
  Otfried Cheong\footnotemark[2]}

\maketitle

\begin{abstract}
  We study the motion-planning problem for a car-like robot whose
  turning radius is bounded from below by one and which is allowed to
  move in the forward direction only (Dubins car).  For two robot
  configurations $\sigma, \sigma'$, let $\ell(\sigma, \sigma')$ be the
  shortest bounded-curvature path from~$\sigma$ to~$\sigma'$.  For $d
  \geq 0$, let $\ell(d)$ be the supremum of $\ell(\sigma, \sigma')$,
  over all pairs $(\sigma, \sigma')$ that are at Euclidean
  distance~$d$.  We study the function~$\dub(d) = \ell(d) - d$, which
  expresses the difference between the bounded-curvature path length
  and the Euclidean distance of its endpoints.  We show that $\dub(d)$
  decreases monotonically from $\dub(0) = 7\pi/3$ to $\dub(\ds) =
  2\pi$, and is constant for $d \geq \ds$.  Here $\ds \approx
  1.5874$. We describe pairs of configurations that exhibit the
  worst-case of $\dub(d)$ for every distance~$d$.
\end{abstract}

\section{Introduction}

\emph{Motion planning} or \emph{path planning} involves computing a
feasible path, possibly optimal for some criterion such as time or
length, of a robot moving among obstacles; see the book by
Lavalle~\cite{l-pa-06} and book chapters by Halperin et
al.~\cite{hkl-04} and Sharir~\cite{s-04}.  A robot generally comes
with physical limitations, such as bounds on its velocity,
acceleration or curvature. Such differential constraints restrict the
geometry of the paths the robot can follow. In this setting, the goal
of motion planning is to find a feasible (or optimal) path satisfying
both global (obstacles) and local (differential) constraints if it
exists.

In this paper, we study the \emph{bounded-curvature} motion planning
problem which models a car-like robot. 
A car (with front-wheel steering) is constrained to move in the direction
that the rear wheels are pointing, and it has a fixed maximum steering
angle. This makes the car travel in a motion with fixed \emph{minimum
turning radius}, which means that the car must follow a
\emph{curvature-constrained} path. More precisely, we have the following
robot model:

\paragraph{Robot model (Dubins car).} 
The robot is considered a rigid body that moves in the plane. A
\emph{configuration} of the robot is specified by both its location, a
point in~$\R^2$ (typically, the midpoint of the rear axle), and its
orientation, or direction of travel.  The robot is
constrained to move in the forward direction, and its turning radius is
bounded from below by a positive constant, which can be assumed to be equal
to one by scaling the space. In this context, the robot follows a
bounded-curvature path, that is, a differentiable curve whose curvature is
constrained to be at most one almost everywhere.

\medskip

Planning the motion of a car-like robot has received considerable
attention in the literature.  
In this paper, we consider the cost of this restriction: How much
longer is the shortest path made by such a robot compared to the
Euclidean distance travelled?   

Formally, consider two configurations $\sigma$ and $\sigma'$. Let
$\ell(\sigma, \sigma')$ denote the length of a shortest
curvature-constrained path from~$\sigma$ to~$\sigma'$, and let
$d(\sigma, \sigma')$ denote the Euclidean distance between $\sigma$
and~$\sigma'$.  We define
\begin{align}
  \label{eq:def-dub}
  \dub(d) & = \sup \{ \ell(\sigma,\sigma') - d \mid
  \text{$\sigma$, $\sigma'$ configurations with $d(\sigma,\sigma') =
    d$}\}.
\end{align}
Note that the supremum here is not a maximum, as the path length is
not a continuous function of the orientations at the two endpoints.
Our goal is to understand the function $\dub: \R \mapsto \R$ in
detail.  While this is a natural and fundamental question related to
motion planning with bounded curvature, it is also a relevant question
that has repeatedly appeared in the literature, with only partial
answers so far.

Dubins~\cite{d-cmlcvcpitpt-57} showed that the shortest
curvature-constrained path between two configurations consists of at most
three segments, each of which is either a straight segment or a circular
arc of radius one.  Using ideas from control theory, Boissonnat et
al.~\cite{bcl-spbcp-94}, in parallel with Sussmann and
Tang~\cite{st-sprsc-91}, gave an alternative proof.
Sussmann~\cite{s-s3dppcb-95} extended the characterization to the
3-dimensional case. Bui et al.~\cite{bsbl-spsnrmf-94} discussed how the
types of optimal paths partition the configuration space, and also proved
that optimal paths for free final orientation have at most two
segments~\cite{bb-arcomfop-94}.  Significant work has been done on the
problems of deciding whether a bounded-curvature path exists between given
configurations among different kinds of obstacles and finding the shortest
such path~\cite{fw-pcm-91, bk-cbtnc-05, rw-ctdccspp-98,
  jc-pspmr-92, aw-aaccsp-01, bk-caasbcp-08,
  art-mpscrmo-95,ablrsw-ccspcp-02, bl-ptacspbcmo-03,
  bgkl-ltaccpbcsp-02}.

At least two interesting problems have been studied where not
configurations but only \emph{locations} for the robot are given.  The
first problem considers a \emph{sequence} of points in the plane, and
asks for the shortest curvature-constrained path that visits the
points in this sequence.  In the second problem, the \emph{Dubins
  traveling salesman problem}, the input is a \emph{set} of points in
the plane, and asks to find a shortest curvature-constrained path
visiting all points.  Both problems have been studied by researchers
in the robotics community, giving heuristics and experimental
results~\cite{sfb-opptspdv-05,mc-rhpdtsp-06,leny_07}. From a
theoretical perspective, Lee et al.~\cite{lcksc-accspsp-00} gave a
linear-time, constant-factor approximation algorithm for the first
problem.  No general approximation algorithms are known for the Dubins
traveling salesman problem (the approximation factor of the known
algorithms depends on the smallest distance between points).

All this work depends on some knowledge of the function~$\dub$.  Lee
et al.~\cite{lcksc-accspsp-00}, for instance, prove that the
approximation ratio of their algorithms is $\max(\A, \pi/2 + \B/\pi)$,
where $\A = 1 + \sup\{\dub(d)/d \mid d \geq 2\}$ and $\B =
\sup\{\dub(d) + d \mid d \leq 2\}$.  They claim without proof that
$\dub(d) \leq 2\pi$ for $d\geq 2$ and derive from this that $\A = 1 +
\pi$ and $\B \leq 5\pi/2 + 3$, leading to an approximation ratio of
about~$5.03$.  We give the first proof of~$\A = 1 + \pi$, and improve
the second bound to~$\B = 2 + 2\pi$, improving the approximation ratio
of their algorithm to~$2 + 2/\pi + \pi/2 \approx 4.21$.

Savla et al.~\cite{sfb-tspdv-08} prove that $\dub(d) \leq \kappa \pi$,
where $\kappa \in [2.657,2.658]$, and conjecture based on numerical
experiments that the true bound is $7\pi/3$.  We show that this is
indeed true.

\paragraph{Results.}  We show that $d \mapsto \dub(d)$ is a decreasing
function with two breakpoints, at~$\sqrt{2}$ and at~$\ds \approx 1.5874$
(see Figure~\ref{fig:dub-graph}). More precisely, we have $\dub(0) =
7\pi/3$ and the two breakpoint values are $\dub(\sqrt{2}) = 5\pi/2 -
\sqrt{2}$, and $\dub(\ds) = 2\pi$. The function $\dub(d)$ is constant and
equal to~$2\pi$ for~$d \geq \ds$.
\begin{figure}[h]
  \centerline{\includegraphics{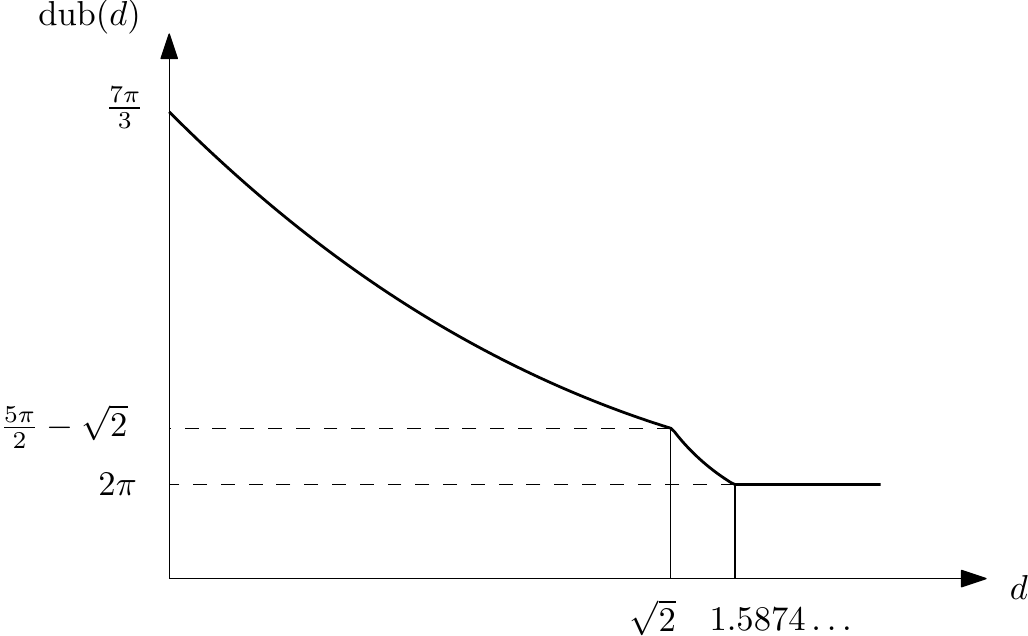}}
  \caption{The graph of the function $\dub(d)$.}
  \label{fig:dub-graph}
\end{figure}

For $0 \leq d < \sqrt{2}$ and for $d \geq \ds$, the supremum
in~\eqref{eq:def-dub} is in fact a maximum, and we give
configurations~$\sigma, \sigma'$ at distance~$d$ such that
$\ell(\sigma, \sigma') = \dub(d) + d$.  Perhaps surprisingly, for
$\sqrt{2} \leq d < \ds$, there are no such configurations---the
supremum is not a maximum.

Our proof is long and contains calculations that some readers may find
tedious.  After laying the necessary groundwork in
Section~\ref{sec:preliminaries}, we will therefore provide only a
high-level proof in Section~\ref{sec:overview}.  We fill in the
details in Sections~\ref{sec:case-c} to~\ref{sec:case-b}, leaving some
of the more technical or tedious calculations to an appendix.

\section{Preliminaries}
\label{sec:preliminaries}

\paragraph{Notations.}

For two points $P$ and~$Q$, we denote by $\overline{PQ}$ the line
segment with endpoints~$P$ and~$Q$, and by $\arc{PQ}$ an arc of unit
radius with endpoints~$P$ and~$Q$. (If the length of $\overline{PQ}$
is less than two then there are four such arcs, so unless it is clear
from the context, we will specify the supporting circle and the
orientation of the arc.)  We denote the length of the segment
$\overline{PQ}$ as $|\overline{PQ}|$ or simply as~$|PQ|$, and the
length of the arc~$\arc{PQ}$ as~$|\arc{PQ}|$.

\medskip

Without loss of generality, we assume that the starting configuration
is $(0, 0, \alpha)$---that is, we start at the origin~$S = (0,0)$ with
orientation~$\alpha$---and the final configuration is~$(d, 0,
\beta)$---that is, we arrive at $F = (d,0)$ with
orientation~$\beta$. Here, $\alpha$ and $\beta$ express the
orientation of the robot as an angle with the positive $x$-axis, and
$d \geq 0$ is the Euclidean distance of the two configurations.

The \emph{open} unit (radius) disks tangent to the
starting and final configurations are denoted~$L_S, R_S, L_F, R_F$,
where the letters $L$ or $R$ depend on whether the disk is located on
the left or right side of the direction vector (see
Figure~\ref{fig:example-configurations}).
\begin{figure}
  \centerline{\includegraphics{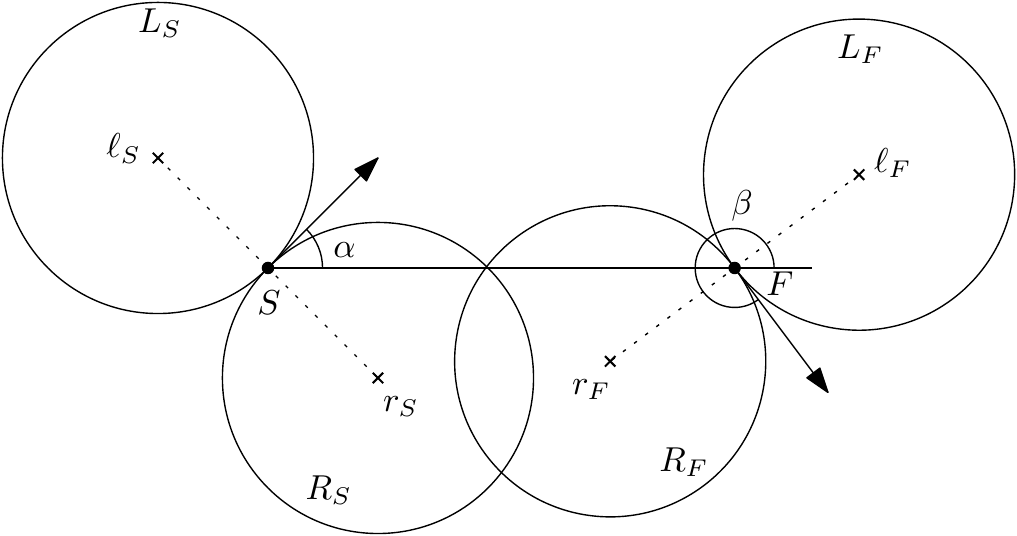}}
  \caption{An example of a pair of configurations.}
  \label{fig:example-configurations}
\end{figure}

Let $\ell_S, r_S, \ell_F, r_F$ denote the centers of $L_S, R_S, L_F,
R_F$, respectively. For future reference, we note their coordinates:
\begin{align*}
  \ell_S & = (\cos(\alpha+\pi/2), \sin(\alpha+\pi/2))  =  (-\sin\alpha,
  \cos\alpha) \\
  r_S & = (\cos(\alpha-\pi/2), \sin(\alpha-\pi/2))  =  (\sin\alpha,
  -\cos\alpha) \\
  \ell_F & = (d+\cos(\beta+\pi/2), \sin(\beta+\pi/2))  =  (d-\sin\beta,
  \cos\beta) \\
  r_F & =  (d+\cos(\beta-\pi/2), \sin(\beta-\pi/2))  =  (d+\sin\beta,
  -\cos\beta).
\end{align*}

\paragraph{Distances between centers.} 

The following distances will be frequently used:
\begin{align}
  \dl & = |\ell_{S}\ell_{F}| =
  \sqrt{(d-\sin\beta+\sin\alpha)^2+(\cos\beta-\cos\alpha)^2}
  \label{eq:dl-alpha-beta}\\
  \dr & = |r_{S}r_{F}| =
  \sqrt{(d+\sin\beta-\sin\alpha)^2+(-\cos\beta+\cos\alpha)^2}
  \label{eq:dr-alpha-beta}\\ 
  \dlr & = |\ell_{S}r_{F}| = 
  \sqrt{(d+\sin\beta+\sin\alpha)^2+(-\cos\beta-\cos\alpha)^2}
  \label{eq:dlr-alpha-beta}\\
  \drl & = |r_{S}\ell_{F}| = 
  \sqrt{(d-\sin\beta-\sin\alpha)^2+(\cos\beta+\cos\alpha)^2}
  \label{eq:drl-alpha-beta}.
\end{align}

\paragraph{Dubins paths.} Dubins~\cite{d-cmlcvcpitpt-57} showed that
for two given configurations in the plane, shortest
bounded-curvature paths consist of arcs of unit radius circles
($C$-segments) and straight line segments ($S$-segments); moreover, such
shortest paths are of type \pccc{} or \pcsc{}, or a substring thereof. These
types of paths are referred to as \emph{Dubins paths}.

For given $d \geq 0$, $\alpha, \beta \in [0, 2\pi]$, there are up to six
types of Dubins paths.  The two path types~\plsl{} and~\prsr{} use outer
tangents---these path types exist for any choice of~$d, \alpha, \beta$.
The two path types~\plsr{} and~\prsl{} use inner tangents, and exist only
when the corresponding disks are disjoint.  In particular, \plsr{} exists
if and only if $\dlr \geq 2$, and \prsl{} exists if and only if $\drl \geq
2$.   The remaining two path types \plrl{} and \prlr{} exist whenever there
is a disk tangent to the two disks, and so \plrl{}  exists if and only if
$\dl \leq 4$, and \prlr{} exists if and only if $\dr \leq 4$.

Dubins showed that in \plrl- and \prlr-paths the middle circular arc has
length larger than~$\pi$. This implies that of the two unit radius disks
tangent to $L_S$ and~$L_F$, only one is a candidate for the middle arc of
an \plrl-path, and similar for \prlr-paths.

\medskip

For $d \geq 0$ and $0 \leq \alpha,\, \beta \leq 2\pi$, we define
$\llsl(d, \alpha, \beta)$ to be the length of the $LSL$-path from $(0,
0, \alpha)$ to $(d, 0, \beta)$.  We define $\lrsr, \llsr, \lrsl,
\llrl, \lrlr$ similarly, defining the length to be~$\infty$ if no path
of that type exists.  The length of the shortest bounded-curvature
path from $S$ to~$F$ is then
\[
\l(d, \alpha, \beta) = \min \big\{\llsl(d, \alpha, \beta), 
\lrsr(d, \alpha, \beta), 
\llsr(d, \alpha, \beta), \lrsl(d, \alpha, \beta), 
\llrl(d, \alpha, \beta), \lrlr(d, \alpha, \beta) \big\},
\]
and our goal is to bound $\dub(d) = \sup_{0 \leq \alpha, \beta \leq
  2\pi} \l(d, \alpha, \beta) - d$. (Note that the supremum here is not
always a maximum as the function~$\l$ is not continuous.)

We will often suppress the argument~$d$ for these functions when the
distance~$d$ is fixed and understood.

\paragraph{Monotonicity of the Dubins cost function.}

Let~$\SQUAREAB = [0, 2\pi]^{2}$ denote the range of~$(\alpha, \beta)$.
Consider two distances~$d_1 < d_2$ and $(\alpha, \beta) \in
\SQUAREAB$, and assume that we have a bounded-curvature path from $(0, 0,
\alpha)$ to~$(d_1, 0, \beta)$ of length~$\ell \leq \dub(d_1) + d_1$.
If this path has a horizontal tangent where the path is oriented to
the right (in the direction of the positive $x$-axis), then we can
insert a horizontal segment of length~$d_2 -d_1$ at this point, and
obtain a path from $(0, 0, \alpha)$ to $(d_2, 0, \beta)$ of
length~$\ell + (d_2 - d_1) \leq \dub(d_1) + d_2$. See, for instance,
Figure~\ref{fig:monotonicity-rsl}(a).
\begin{figure}
  \centerline{\subfigure[$\ell(d_2, \alpha, \beta) 
      \leq \lrsl(d_1, \alpha, \beta) +(d_2
      -d_1)$.]{\includegraphics{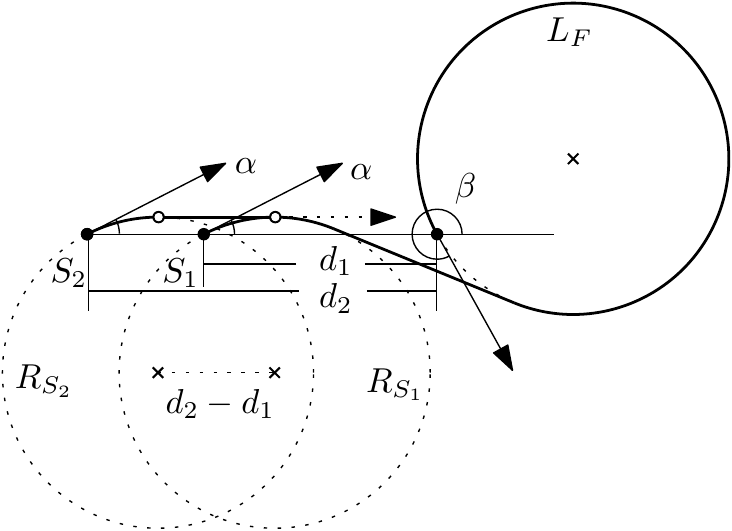}} 
    \hspace{1cm}
    \subfigure[An \prlr-path that does not have a right horizontal
      tangent.]{\includegraphics{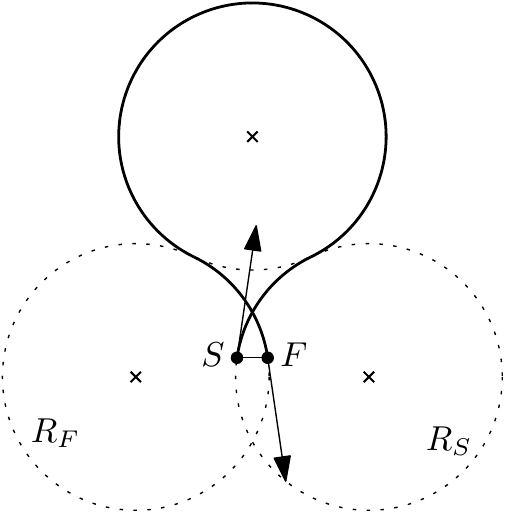}}} 
  \caption{(a) Monotonicity of the length function of \pcsc-paths;
    (b) Non-monotonicity of the length function of \pccc-path.}
  \label{fig:monotonicity-rsl}
\end{figure}

If this was possible for all $(\alpha, \beta) \in \SQUAREAB$, then it
would imply that $\dub(d_2) \leq \dub(d_1)$, and it would follow that
the Dubins cost function is monotone.  Unfortunately, not all Dubins
paths have horizontal tangents with the correct orientation (see
Figure~\ref{fig:monotonicity-rsl}(b) for an example), and so proving
the monotonicity of the Dubins cost function will require much more
work.  However, we can start with the following lemma:
\begin{lemma}
  \label{lem:monotonicity-csc}
  Let $d_1 < d_2$, and $(\alpha, \beta) \in \SQUAREAB$.  If there is a
  path of length~$\ell$ of type \prsr, \plsl, \plsr, or \prsl{}
  from~$(0, 0, \alpha)$ to~$(d_1, 0, \beta)$, then there is a path of
  length~$\ell + (d_2 - d_1)$ from~$(0, 0, \alpha)$ to~$(d_2, 0,
  \beta)$.
\end{lemma}
\begin{proof}
  It suffices to show that any of these path types must have a
  horizontal tangent oriented in the positive $x$-direction.  By
  symmetry, it suffices to show this for \prsr- and \prsl-paths.  The
  topmost point on a \prsr-path necessarily has the correct tangent,
  so consider an \prsl-path.  It consists of a right-turning
  arc~$\arc{ST_1}$ on~$R_S$, a segment~$\overline{T_1T_2}$, and a
  left-turning arc~$\arc{T_2F}$ on~$L_F$.  If $\arc{ST_1}$ contains
  the topmost point of~$R_S$, or if~$\arc{T_2F}$ contains the
  bottommost point of~$L_F$, these have the correct tangent, and we
  are done.  Otherwise the path cannot possibly reach the positive
  $x$-axis, a contradiction.
\end{proof}

\paragraph{Symmetries.}

For a fixed $d \geq 0$, determining $\dub(d)$ essentially amounts to
finding $(\alpha, \beta) \in \SQUAREAB$ maximizing $\l(\alpha, \beta)$
(``essentially'' since the maximum may not actually be assumed).  We
now observe that the function $\l(\alpha, \beta)$ has two symmetries.

First, we can mirror a path around the $x$-axis. This maps $\alpha$ to
$-\alpha$, $\beta$ to $-\beta$, left disks to right disks, and right
disks to left disks.  As a result, we have, say, $\llsr(d,
\alpha,\beta) = \lrsl(d, -\alpha, -\beta)$, and in general we have
$\l(d, \alpha, \beta) = \l(d, -\alpha, -\beta)$. See
Figure~\ref{fig:x-mirror}.

Second, we can mirror a path around the line $x = d/2$ and reverse the
direction of the path.  If the original path connected $(0,0,\alpha)$
with $(d, 0, \beta)$, the new path connects $(0, 0, 2\pi - \beta)$ to
$(d, 0, 2\pi - \alpha)$.  The transformation maps left disks to left
disks and right disks to right disks, so we have, for instance
$\llsr(d, \alpha, \beta) = \llsr(d, 2\pi - \beta, 2\pi - \alpha)$, and
in general $\l(d, \alpha, \beta) = \l(d, 2\pi - \beta, 2\pi -
\alpha)$. See Figure~\ref{fig:vertical-mirror}.
\begin{figure}
  \centerline{\subfigure[Original path]{\includegraphics{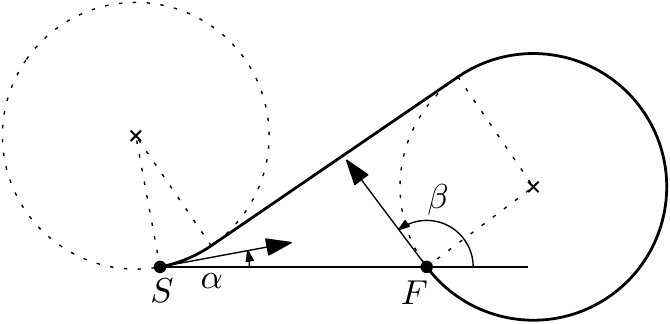}}
    \hspace{1cm}
    \subfigure[Mirrored along $x$-axis]{\includegraphics{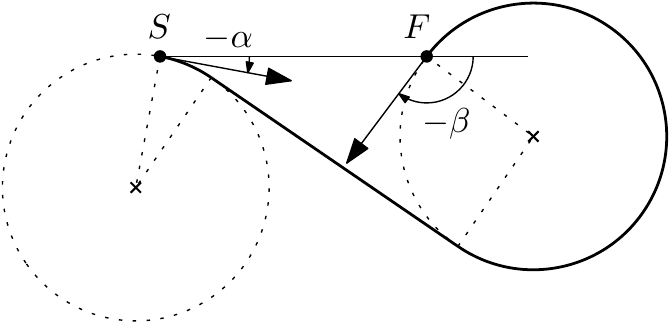}}}
  \centerline{\subfigure[Point symmetry in $(\pi,\pi)$]{%
      \includegraphics{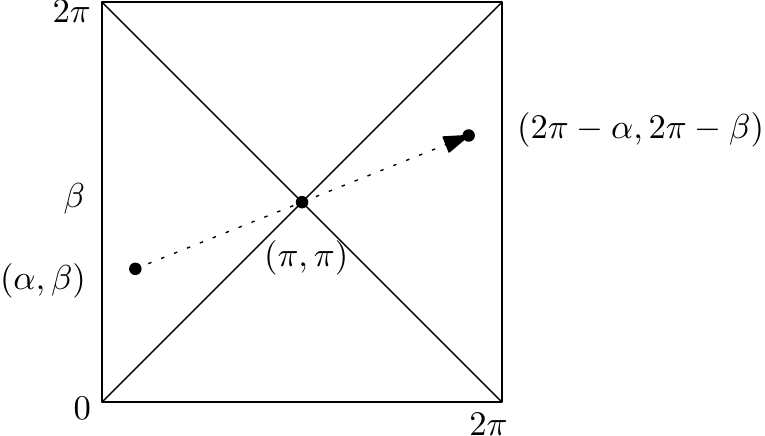}}}
  \caption{Symmetry by mirroring the path along the $x$-axis.}
  \label{fig:x-mirror}
\end{figure}

\begin{figure}
  \centerline{\subfigure[Original path]{\includegraphics{figg/original}}
    \hspace{1cm}
    \subfigure[Mirrored along the line $x=d/2$ while reversing the
      direction]{\includegraphics{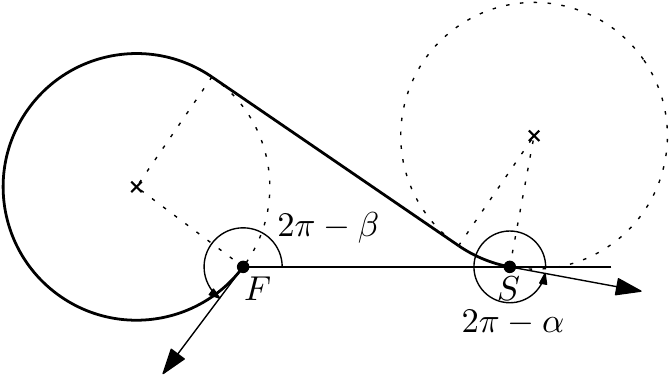}}}
  \centerline{\subfigure[Reflection around the line $\beta=2\pi-\alpha$]{%
      \includegraphics{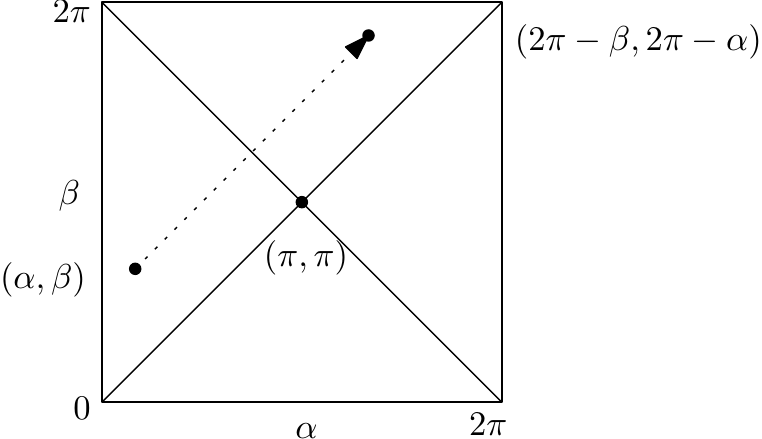}}}
  \caption{Symmetry by mirroring the path along the line $x=d/2$ while
    reversing the direction.}
  \label{fig:vertical-mirror}
\end{figure}

Considered as symmetries on~$\SQUAREAB$, the mapping $(\alpha, \beta)
\mapsto (-\alpha, -\beta)$ is a point symmetry in $(\pi,\pi)$, while
the mapping $(\alpha, \beta) \mapsto (2\pi - \beta, 2\pi - \alpha)$ is
a reflection around the line $\beta = 2\pi - \alpha$.

It follows that $\sup_{(\alpha, \beta) \in [0, 2\pi]^{2}}\ell(d,
\alpha, \beta) = \sup_{(\alpha, \beta) \in \Delta} \ell(d, \alpha,
\beta)$, where $\Delta$ is the triangle with corners $(0,0)$,
$(\pi,\pi)$, and $(0,2\pi)$, or in other words the region
\begin{align*}
  \Delta: \quad 0 \leq \alpha \leq \pi \quad \text{and} 
  \quad \alpha \leq \beta \leq 2\pi - \alpha.  
\end{align*}
In the following we will thus be able to restrict our considerations
to the triangle~$\Delta$ (see Figure~\ref{fig:Delta-alpha-beta}).

\paragraph{A new parameterization.}

We now introduce a new parameterization of the $(\alpha,
\beta)$-plane, which will sometimes be more convenient to work with:
\begin{align*}
\sigma = \frac{\beta+\alpha}{2},\quad \delta = \frac{\beta-\alpha}{2}.  
\end{align*}
In other words, we have 
\begin{align*}
\alpha = \sigma - \delta, \quad \beta = \sigma + \delta.
\end{align*}

Recall our triangle $\Delta$ from above. In the
$(\sigma,\delta)$-representation, the triangle~$\Delta$ is the
triangle
\begin{align*}
  \Delta: \quad 0 \leq \sigma \leq \pi \quad 
  \text{and} \quad 0\leq \delta\leq \sigma,
\end{align*}
or the bottom right half of the square~$\SQUARE = [0, \pi]^{2}$ (see
Figure~\ref{fig:Delta-sigma-delta}).  In this representation, our
first symmetry maps $(\sigma, \delta)$ to $(-\sigma, -\delta)$, while
the second symmetry maps $(\sigma, \delta)$ to $(-\sigma, \delta)$.
We thus have point symmetry in the origin, as well as mirror symmetry
around the $\delta$-axis.  In addition, $(\sigma+\pi, \delta+\pi)$
represents the same angles as $(\sigma, \delta)$ since
$\alpha=\sigma-\delta$ and $\beta=\sigma+\delta$, and so we also have
point symmetry in the point $(\pi/2,\pi/2)$, or in other words
$\ell(d, \sigma,\delta) = \ell(d, \pi-\sigma,\pi-\delta)$.

\begin{figure}
  \centerline{\subfigure[$(\alpha, \beta)$-representation]{%
      \includegraphics{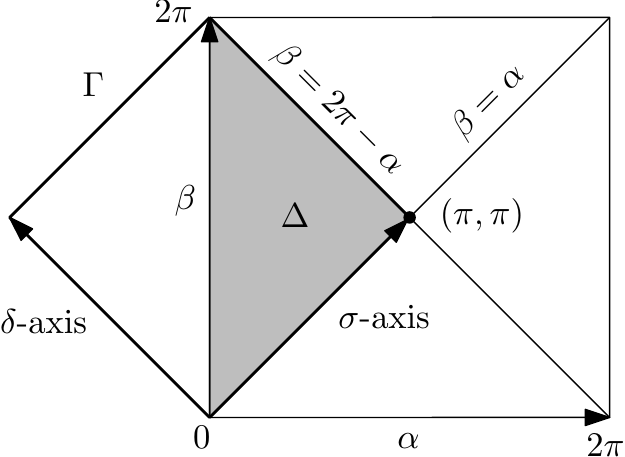}%
      \label{fig:Delta-alpha-beta}}
    \hspace{2cm}
    \subfigure[$(\sigma, \delta)$-representation]{%
      \includegraphics{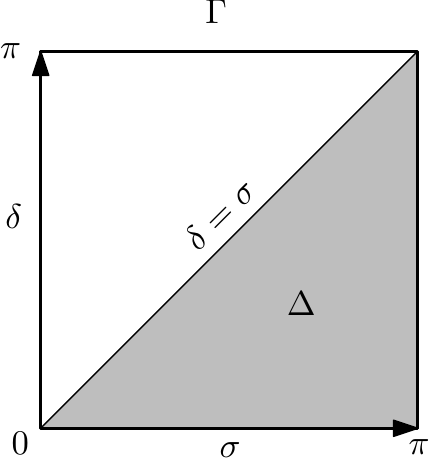}%
      \label{fig:Delta-sigma-delta}}}
  \caption{$\Delta$ in two different representations.}
  \label{fig:Delta}
\end{figure}

\paragraph{Distances between centers (using $\sigma$ and
  $\delta$).}

The following lemma will allow us to express the center distances in terms
of $\sigma$ and~$\delta$:
\begin{lemma}
  \label{lem:equal-angle}
  Let $d \geq 0$, and define the points $p(\Th) = (\cos \Th, \sin
  \Th)$ and $q(\Th) = (d - \cos \Th, \sin\Th)$ on the unit radius
  disks around $(0,0)$ and~$(d, 0)$.  Then
  \begin{equation*}
    |p(\Th-\phi)q(\Th+\phi)|^{2} = 
    d^{2} - 4d \cos\Th\cos\phi + 4
    \cos^{2}\Th.
  \end{equation*}
\end{lemma}
\begin{proof}
  We have
  \begin{align*}
    |p(\Th-\phi) q(\Th + \phi)|^{2}  & =
    \big(\cos(\Th - \phi) - d + \cos(\Th + \phi)\big)^{2} +
    \big(\sin(\Th - \phi) - \sin(\Th + \phi)\big)^{2} \\
    & = (\cos\Th\cos\phi + \sin\Th\sin\phi - d 
    + \cos\Th\cos\phi - \sin\Th\sin\phi)^{2} \\
    & \relphantom{=} {} + (\sin\Th\cos\phi - \cos\Th\sin\phi 
    - \sin\Th\cos\phi - \cos\Th\sin\phi)^{2} \\
    & = (2\cos\Th\cos\phi - d)^{2} + 4\cos^{2}\Th\sin^{2}\phi \\
    & = d^{2} - 4d\cos\Th\cos\phi + 4\cos^{2}\Th\cos^{2}\phi 
    + 4\cos^{2}\Th\sin^{2}\phi \\
    & = d^{2} - 4d \cos\Th\cos\phi + 4 \cos^{2}\Th.\qedhere
  \end{align*}
\end{proof}

Lemma~\ref{lem:equal-angle} leads to the following expressions for the
squared distances between our disk centers:
\begin{align}
  \dl^{2} & = |\ell_S\ell_F|^2 = d^{2} - 4d\sin\delta\cos\sigma +
4\sin^{2}\delta \label{eq:dl-sigma-delta}\\
  \dr^{2} & = |r_Sr_F|^2 = d^{2} + 4d\sin\delta\cos\sigma + 4\sin^{2}\delta
	\label{eq:dr-sigma-delta}\\
  \dlr^{2} & = |\ell_Sr_F|^2  = d^{2} + 4d\cos\delta\sin\sigma +
4\cos^{2}\delta \label{eq:dlr-sigma-delta}\\
  \drl^{2} & = |r_S\ell_F|^2  = d^{2} - 4d\cos\delta\sin\sigma +
4\cos^{2}\delta \label{eq:drl-sigma-delta}
\end{align}  
To see this, observe that $\ell_{S} = p(\alpha+\pi/2)$, $r_S =
p(\alpha-\pi/2)$, $\ell_{F} = q(\pi/2-\beta)$, and $r_F =
q(3\pi/2-\beta)$.  For~$\dl$, set $\Th = \pi/2-\delta$,
$\phi=-\sigma$; for~$\dr$, set $\Th = \pi/2-\delta$,
$\phi=\pi-\sigma$; for~$\dlr$, set $\Th = \pi-\delta$,
$\phi=\pi/2-\sigma$; for~$\drl$, set $\Th = -\delta$,
$\phi=\pi/2-\sigma$.

\paragraph{The case $d = 0$.}

We first argue that $\dub(0) = 7\pi/3$.  The case $d = 0$ is much easier
since there is only one degree of freedom: Without loss of generality we
can assume $\alpha = 0$.  It is easy to verify that for any~$\beta$ there
is a \pccc-path of length at most $7\pi/3$. For $\beta = \pi$, no Dubins
path has length shorter than $7\pi/3$, and so~$\dub(0)= 7\pi/3$ (see
Figure~\ref{fig:worst-length}).  In the rest of this paper we can therefore
mostly assume~$d > 0$, and avoid some degeneracies.
\begin{figure}
  \centerline{\includegraphics{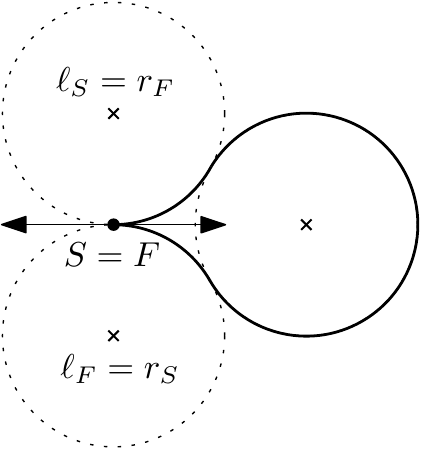}}
  \caption{For $d=0$, $\alpha=0$ and $\beta=\pi$, we have
    $\llrl(d,\alpha,\beta)=\lrlr(d, \alpha, \beta)={7\pi}/{3}$.}
  \label{fig:worst-length}
\end{figure}

\section{The overall proof}
\label{sec:overview}

We subdivide the pairs of orientations $(\alpha, \beta) \in \SQUARE$
into three cases.  The case distinction is based on the existence of
the $\plsr$-path and the~$\prsl$-path:
\begin{denseitems}
\item When neither the $\plsr$-path nor the $\prsl$-path exists, then
  we are in case~A;
\item when the $\prsl$-path exists, but the $\plsr$-path does not
  exist, then we are in case~B; and
\item when the $\plsr$-path exists, then we are in case~C.
\end{denseitems}
The reader may wonder why the case distinction is not symmetric. We
already broke the symmetry when we restricted our investigation to the
triangle~$\Delta$. 

The $\plsr$-path exists if and only if the disks~$L_S$ and~$R_F$ are
disjoint (since the disks are open, they may touch, but cannot
overlap), or, equivalently, if $|\l_S r_F| = \dlr(\alpha, \beta) \geq
2$.  The $\prsl$-path exists if and only if~$R_S$ and~$L_F$ are
disjoint, that is if $|r_S \ell_F| = \drl(\alpha, \beta) \geq 2$. 
We can thus define three regions of the square~$\SQUARE$ for the three
cases: 
\begin{align*}
  A &= \{ (\alpha, \beta) \in \SQUARE \mid \dlr(\alpha, \beta) < 2
  \;\text{and}\; \drl(\alpha, \beta) < 2\},\\
  B &= \{ (\alpha, \beta) \in \SQUARE \mid \dlr(\alpha, \beta) < 2
  \;\text{and}\; \drl(\alpha, \beta) \geq 2\},\\
  C &= \{ (\alpha, \beta) \in \SQUARE \mid \dlr(\alpha, \beta) \geq 2\}.
\end{align*}
Figure~\ref{fig:regions} shows the subdivision of~$\SQUARE$ into the
three regions for three different distances~$d$.  
\begin{figure}
  \centering
  \subfigure[$d=0.4$]{
    \includegraphics{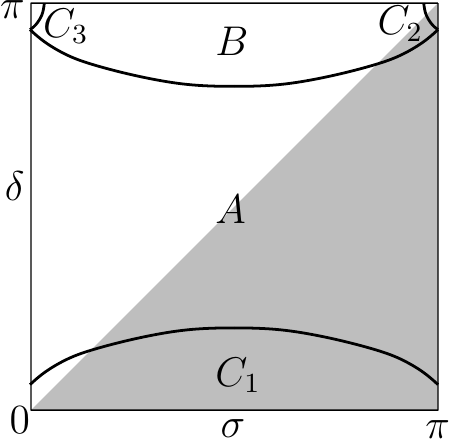}
  }
  \subfigure[$d=1.4$]{
    \includegraphics{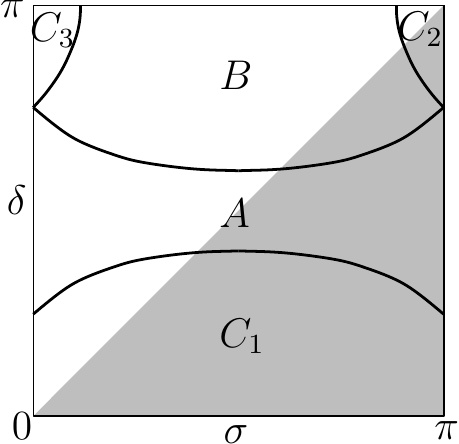}
  }
  \subfigure[$d=2.4$]{
    \includegraphics{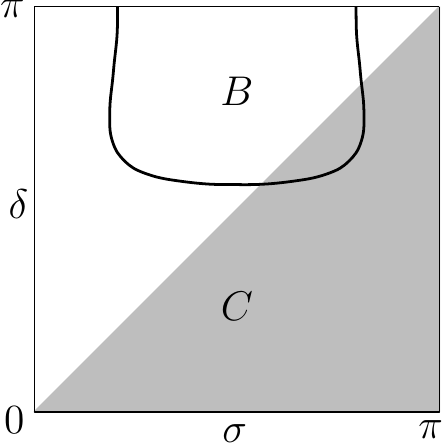}
  }
  \caption{Regions of cases~A, B and C in $\SQUARE$}
  \label{fig:regions}
\end{figure}
It will be
convenient to also define the parts of $A$, $B$ and $C$ that lie
inside the triangle~$\Delta$:
\begin{align*}
  \AD = A \cap \Delta, \quad 
  \BD = B \cap \Delta, \quad \text{and} \quad
  \CD = C \cap \Delta.
\end{align*}
We can now restrict the Dubins cost function~$\dub(d)$ to the three
regions
\begin{align*}
\dub_A(d) & = \sup\nolimits_{(\alpha, \beta) \in \AD} \l(d,\alpha,
\beta) - d\\
\dub_B(d) & = \sup\nolimits_{(\alpha, \beta) \in \BD} \l(d,\alpha,
\beta) -d\\
\dub_C(d) & = \sup\nolimits_{(\alpha, \beta) \in \CD} \l(d,\alpha,
\beta) - d\\
\intertext{and we have}
\dub(d) & = \max \{\, \dub_A(d),\; \dub_B(d),\; \dub_C(d) \,\}.
\end{align*}

\paragraph{Case~C.}

This is the easiest case, and we discuss it first in
Section~\ref{sec:case-c}.  Since both \prsr- and \plsr-paths exist, we
can show that for $(\alpha, \beta) \in \CD$ the \prsr-path or the
\plsr-path has length at most~$d+2\pi$
(Lemma~\ref{lem:dub_c_upper_bound}).  This implies that $\dub_C(d)
\leq 2\pi$.  The bound is tight, as the shortest path from
configuration~$(0,0,\pi)$ to~$(d, 0, \pi)$ has length~$d + 2\pi$
(Lemma~\ref{lem:lower-bound}), and so $\dub_C(d) = 2\pi$.  This
implies a lower bound on the Dubins cost function: $\dub(d) \geq
\dub_C(d) = 2\pi$.

\paragraph{Case~A.}

We first observe that case~A occurs only for $d < 2$, as we have $8 >
\dlr^{2}+ \drl^{2} = 2d^{2} + 8\cos^{2}\delta \geq 2d^{2}$.  In this
case, neither~$\plsr$- nor~$\prsl$-paths exist.  Since 
$|\l_S\l_F| \leq |\l_SS| + |SF| + |F\l_F| = d + 2 < 4$ and
$|r_Sr_F| \leq |r_SS| + |SF| + |Fr_F| = d + 2 < 4$, both 
$\plrl$-path and $\prlr$-path exist everywhere in~$\SQUARE$ for~$d <
2$. 

It follows that the path types in case~A are $\prsr$, $\plsl$,
$\plrl$, and $\prlr$. It turns out that the shortest path is always of
type~$\pccc$ (Lemma~\ref{lem:region2-lrl-rlr-shorter}), and so we can
concentrate entirely on comparing $\llrl(\alpha,\beta)$
and~$\lrlr(\alpha, \beta)$.  To this end, we derive explicit
expressions for the length of~$\pccc$-paths in
Section~\ref{sec:ccc-paths}.  We examine the derivatives of these
functions and show that they are monotone in $\sigma$- and
$\delta$-directions (Lemma~\ref{lem:l-r-fixed}).  Using monotonicity
and Lagrange-multipliers, we show that $\dub_A(d)
=\sup_{(\sigma,\delta)\in A}\ell(d, \sigma,\delta)-d$ is realized at a
point $(\sigma_A, \delta_A)$ on the boundary of the region~$A$
(Lemma~\ref{lem:max-rectii}).

For $0 < d < \sqrt{2}$, this point is the unique point~$(\pi,
\delta_A) \in A$ where $\llrl(\pi,\delta_A) = \lrlr(\pi,\delta_A)$
(Lemma~\ref{lem:a-less-sqrt2}).  For $\sqrt{2} \leq d < 2$, however,
the supremum $\sup_{(\sigma,\delta)\in A}\ell(d, \sigma,\delta)$
occurs on the common boundary of regions~$A$ and~$B$. Since this
boundary belongs to region~$B$, we can conclude that for $\sqrt{2} < d
< 2$ we have $\dub_A(d) \leq \max\{2\pi, \dub_B(d)\}$
(Lemma~\ref{lem:a-larger-sqrt2}).

\paragraph{Case~B.}

In case~B, the \prsl-path exists, while the \plsr-path does not exist.
We have $|\l_S\l_F| \leq |\l_S r_F| + |r_F\l_F| = \dlr+2 < 4$, and so
the $\plrl$-path exists.  We show that the shortest Dubins path is
either the \plrl-path or the \prsl-path
(Lemma~\ref{lem:region3-lrl-rsl-shorter}).

The following lemma shows that case~B occurs only in the upper half of
the square~$\SQUARE$: 
\begin{lemma}
  \label{lem:case-b-basics}
  A point $(\sigma, \delta) \in B$ has $\delta > \pi/2$, and
  for $(\alpha, \beta) \in \BD$ we have
  \begin{align}
    0 \leq \alpha \leq \pi/2 \quad \text{and} 
    \quad \pi+\alpha < \beta < 2\pi-\alpha.
  \end{align}
\end{lemma}
\begin{proof}
  $(\sigma, \delta) \in B$ means $\dlr^{2} < 4$ and $\drl^{2} \geq 4$,
  so $0 > \dlr^{2} - \drl^{2} = 8d\cos\delta\sin\sigma$.  Since
  $\sin\sigma \geq 0$, we must have $\cos\delta < 0$, and thus $\delta
  > \pi/2$.  So $(\sigma,\delta)$ lies in the top half of~$\SQUARE$.
  This top half intersects the triangle~$\Delta$ in the triangle with
  corners (in $(\alpha,\beta)$-coordinates) $(0, \pi)$, $(\pi/2,
  3\pi/2)$, and $(0, 2\pi)$, implying the claim.
\end{proof}

The function~$\lrsl(\alpha, \beta)$ was studied by Goaoc et
al.~\cite{gkl-bcsp-2010}, who gave an explicit expression for its
derivative.  We exploit this to show monotonicity of~$\lrsl(\alpha,
\beta)$ (Lemma~\ref{lem:rsl-changes-alpha-beta}).  

For $0 < d < \sqrt{2}$, monotonicity of $\lrsl(\alpha, \beta)$ easily
implies the following proposition:
\begin{proposition}
  \label{prop:dub-cost-less-sqrt2}
  For $0 < d < \sqrt{2}$, $\dub(d) = \dub_A(d)$. The function
  decreases monotonically from $\dub(0) = 7\pi/3$ to $\dub(\sqrt{2}) =
  5\pi/2 - \sqrt{2}$.
\end{proposition}
\begin{proof}
  Let $0 < d < \sqrt{2}$.  By Lemma~\ref{lem:a-less-sqrt2},
  $\dub_{A}(d) > 2\pi + 2\as - d$, and by the last statement of
  Lemma~\ref{lem:rsl-changes-alpha-beta} we have $\dub_B(d) \leq 2\pi
  + 2\as - d$ (here, $\as = \arcsin(d/2)$).
  Lemmas~\ref{lem:a-less-sqrt2} and~\ref{lem:a-monotonicity} imply
  that $\dub_A(d)$ decreases monotonically from $\dub_A(0) = 7\pi/3$
  to $\dub_A(\sqrt{2}) = 5\pi/2 - \sqrt{2}$, and so $\dub_A(d) > 2\pi
  = \dub_C(d)$.
\end{proof}
For $\sqrt{2}\leq d < 2$, we are able to show that under the
assumption that $\dub_B(d) > 2\pi$, the value~$\dub_B(d)$ is assumed
at the (unique) point $(\alpha_B, \beta_B)$ on the common boundary
of~$\BD$ and~$\CD$ where $\lrsl(\alpha_B, \beta_B)=\llrl(\alpha_B,
\beta_B)$ (Lemma~\ref{lem:dub_b}).  This common boundary, however, is
part of the region~$\CD$.  A shorter \plsr-path exists in~$\CD$, and
in fact there is no $(\alpha, \beta) \in \BD$ with $\l(d,\alpha,\beta)
= \dub_B(d) + d$.  We show that $\dub_B(d)$ is a monotonically
decreasing function (Lemma~\ref{lem:dub-b-sqrt2-2}), starting with
$\dub_B(\sqrt{2}) = \dub_A(\sqrt{2}) = 5\pi/2 - \sqrt{2}$.  Continuity
and monotonicity imply that there is a distance~$\ds$ where
$\dub_B(\ds) = 2\pi$.  We numerically computed $\ds \approx 1.5874$,
and have the following proposition.
\begin{proposition}
  \label{prop:dub-cost-sqrt2-2}
  For $\sqrt{2} \leq d < \ds$, $\dub(d) = \dub_B(d)$. The
  function~$\dub(d)$ decreases monotonically from $\dub_B(\sqrt{2}) =
  5\pi/2 - \sqrt{2}$ to $\dub(\ds) = 2\pi$.  For $\ds \leq d < 2$,
  $\dub(d) = 2\pi$.
\end{proposition}
\begin{proof}
  By Lemma~\ref{lem:a-larger-sqrt2}, we have $\dub_{A}(d) \leq
  \max\{2\pi, \dub_B(d)\} = \max\{\dub_C(d), \dub_{B}(d)\}$, and so
  $\dub(d) = \max\{2\pi, \dub_B(d)\}$.  By
  Lemma~\ref{lem:dub-b-sqrt2-2} we have $\dub_{B}(d) > 2\pi$ for
  $\sqrt{2} \leq d < \ds$.
\end{proof}

\paragraph{The case $d \geq 2$.}

Propositions~\ref{prop:dub-cost-less-sqrt2}
and~\ref{prop:dub-cost-sqrt2-2} describe the function~$\dub(d)$ for $0
\leq d < 2$.  It remains to prove that $\dub(d) = 2\pi$ for $d \geq
2$, which amounts to analyzing $\dub_B(d)$. As shown in
Figure~\ref{fig:regions}, the region~$B$ has a rather different shape
for $d > 2$, and our previous arguments do not carry over without
additional tedious calculations.

Fortunately, monotonicity comes to the rescue.  By
Proposition~\ref{prop:dub-cost-sqrt2-2} we have $\dub(d) = 2\pi$ for
$\ds \leq d < 2$.  We had seen in Lemma~\ref{lem:monotonicity-csc}
that monotonicity holds when the original path is a \pcsc-path.  In
Lemma~\ref{lem:monotonicity-all-but-rlr} we extend this lemma by
including the \plrl-path, at least when the configuration is in
case~$\BD$.  This allows us to prove the claim for $d \geq 2$.
\begin{proposition}
  \label{lem:dub-cost-larger-2}
  For $d \geq 2$, we have $\dub(d) = 2\pi$.
\end{proposition}
\begin{proof}
  Let $d \geq 2$ and $(\alpha, \beta) \in \Delta$.  We need to show
  that $\ell(d, \alpha, \beta) \leq 2\pi + d$.  If $(\alpha, \beta)
  \in \CD$, then this follows from
  Lemma~\ref{lem:dub_c_upper_bound}. Otherwise we must have $(\alpha,
  \beta) \in \BD$, and we have $\delta > \pi/2$ by
  Lemma~\ref{lem:case-b-basics}.  We choose $d_1 < 2$ such that $d_1 >
  2\sin(\pi - \delta)$ and $d_1 > \ds$ and consider the configuration
  $(\alpha, \beta)$ for distance~$d_1$. Since case~A occurs only
  within the $\delta$-range $\arcsin(d_1/2) \leq \delta \leq \pi -
  \arcsin(d_1/2)$ (Lemma~\ref{lem:region-boundaries}), we must be in
  either case~B or case~C, so there is 
  a path of type \prsr, \plsr, \prsl, or~\plrl{} of length at most
  $d_1 + \dub(d_1) = d_1 + 2\pi$ from $(0, 0, \alpha)$ to~$(d_1, 0,
  \beta)$.  By Lemmas~\ref{lem:monotonicity-csc}
  and~\ref{lem:monotonicity-all-but-rlr} there is then a path from
  $(0, 0, \alpha)$ to~$(d, 0, \beta)$ of length at most $d_1 + 2\pi +
  (d - d_1) = d+ 2\pi$.
\end{proof}

\paragraph{Main result.}

We summarize some of our results in the following theorem:
\begin{theorem}
  \label{thm:main}
  The function $\dub(d)$ has two breakpoints at $\sqrt{2}$ and $\ds
  \approx 1.5874$.  For $d < \sqrt{2}$, $\dub(d) = \dub_A(d) \leq
  \dub_A(0) ={7\pi}/{3}$. For $\sqrt{2} \leq d < \ds$, $\dub(d) =
  \dub_B(d) \leq \dub_B(\sqrt{2})={5\pi}/{2}-\sqrt{2}$. For $d
  \geq \ds$, we have $\dub(d) = 2\pi$.
\end{theorem}

\section{Case~C}
\label{sec:case-c}

Both the \prsr-path and the \plsr-path exist in case~C, and we show
that at least one of them has length at most~$d + 2\pi$.  The
arguments presented here are already in Kim's master
thesis~\cite{janghwan}.
\begin{lemma}
  \label{lem:dub_c_upper_bound}
  For $(\alpha, \beta) \in \CD$, we have $\l(d, \alpha, \beta) \leq
  2\pi + d$.
\end{lemma}
\begin{proof}
  Since $L_S \cap R_F = \emptyset$, the starting point $S = (0, 0)$
  does not lie in~$R_F$, and so there is a tangent to~$R_F$
  through~$S$ that touches~$R_F$ from above.  Let $\alpha_\tSR$ be the
  angle made by this tangent and the positive $x$-axis (see
  Figure~\ref{fig:sr}).
  \begin{figure}
    \centerline{\subfigure[Definition of
        $\alpha_\tSR$]{\includegraphics{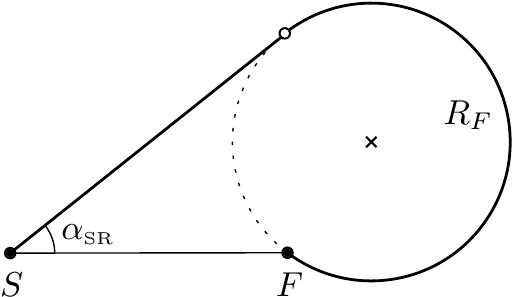}%
        \label{fig:sr}}
      \hspace{2cm}
      \subfigure[\prsr-path]{\includegraphics{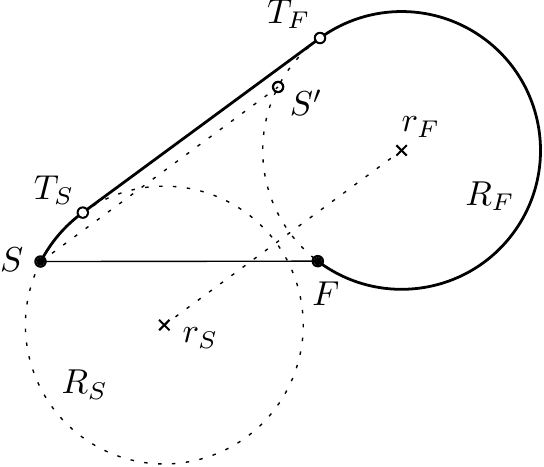}}}
    \caption{\prsr-paths in case~C.}
    \label{fig:sr_rsr}
  \end{figure}

  Let us first assume that $\alpha \geq \alpha_\tSR$, and consider the
  \prsr-path from~$S$ to~$F$.  It consists of an initial right-turning
  arc $\arc{ST_S}$, a straight line segment $\overline{T_S T_F}$, and
  a final right-turning arc $\arc{T_FF}$, where the segment is tangent
  to $R_S$ and $R_F$ at the points $T_S$ and~$T_F$. (When
  $\alpha=\alpha_{\tSR}$ we have $S=T_S$.) See
  Figure~\ref{fig:sr_rsr}.

  Let $t$ be the vector $t = \overrightarrow{T_S T_F}$, and let $S' =
  S + t$. Since $t = \overrightarrow{r_Sr_F}$, we have $R_F = R_S +
  t$, and so $S'$ lies on~$R_F$.  We claim that $S'$ lies on the
  clockwise arc $\arc{F T_F}$.  Indeed, any point~$(\alpha, \beta) \in
  \Delta$ satisfies $\beta \leq \alpha\leq 2\pi-\beta$, which implies
  that $t$ has a positive $y$-component.

  It follows that the length of the \prsr-path is $|\arc{ST_S}| + |T_ST_F| +
  |\arc{T_FF}| = |SS'| + |\arc{S'F}|$. By the triangle-inequality,
  $|SS'| \leq |SF| + |\arc{FS'}|$, and so the length of the \prsr-path
  is at most $|SF| + 2\pi = d + 2\pi$.

  Consider now the case where $\alpha < \alpha_\tSR$. We show that the
  \plsr-path from $S$ to~$F$ has length at most $d + 2\pi$.  The
  \plsr-path consists of an initial left-turning arc $\arc{ST_S}$, a
  straight line segment $\overline{T_S T_F}$ and a final right-turning
  arc $\arc{T_F F}$, where the segment is tangent to $L_S$ and $R_F$
  at points $T_S$ and~$T_F$. See Figure~\ref{fig:lsr_path}.
  \begin{figure}
    \centerline{\subfigure[]{\includegraphics{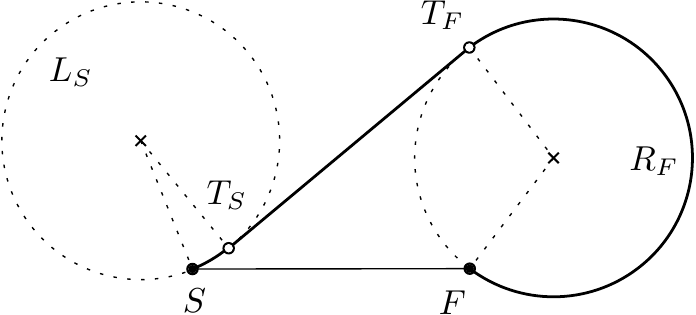}%
	\label{fig:lsr_path}}
      \hspace{2cm}
      \subfigure[]{\includegraphics{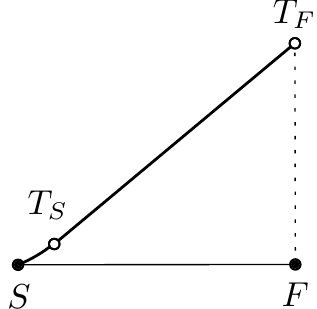}%
	\label{fig:lsr_path_length}}}
    \caption{\plsr-paths in case~C.}
  \end{figure}
  Here, it suffices to observe that $|\arc{ST_S}| + |T_ST_F| \leq |SF|
  + |FT_F|$ (a convex curve contained within another convex curve),
  and so $|\arc{ST_S}| + |T_ST_F| + |\arc{T_FF}| \leq |SF| + |FT_F| +
  |\arc{T_FF}| \leq |SF| + |\arc{FT_F}| + |\arc{T_FF}| = d + 2\pi$.
\end{proof}
It turns out that the bound in Lemma~\ref{lem:dub_c_upper_bound} is
tight, and we obtain:
\begin{lemma}
  \label{lem:lower-bound}
  For any $d > 0$ we have $\dub_C(d) = 2\pi$.
\end{lemma}
\begin{proof}
  Lemma~\ref{lem:dub_c_upper_bound} implies that $\dub_C(d) \leq
  2\pi$, so it remains to provide a matching lower bound.  We will
  show that $\l(d, \pi, \pi) = 2\pi + d$, and since $(\pi,\pi) \in
  \CD$, this proves the claim.  Consider a shortest bounded-curvature
  path~$G$ from $(0,0,\pi)$ to~$(d, 0, \pi)$.  This path must
  intersect the line~$x = 0$ in a point~$p$ and the line~$x = d$ in a
  point~$q$.  The distance $|pq|$ is at least~$d$.  If the path from
  $S$ to~$p$ intersects $L_S \cup R_S$, then Ahn et
  al.~\cite[Fact~1]{lcksc-accspsp-00} showed that it has length at
  least~$\pi$.  Otherwise the path avoids $L_S\cup R_S$ and hence must
  have length at least~$\pi$. The same argument applies to the path
  from~$q$ to~$F$, and so the total length of~$G$ is at least~$d +
  2\pi$.
\end{proof}

Since $\dub(d) \geq \dub_C(d)$, this establishes a lower bound for the
Dubins cost function.

\section{Regions of the square $\SQUARE$ for $0 < d < 2$}
\label{sec:regions}

Before we can discuss cases~A and~B in detail, we need to have a
precise description of the regions~$A$ and~$B$ of the
square~$\SQUARE$.  As we saw in Section~\ref{sec:overview}, it
suffices to do this for~$0 < d < 2$.

We define the angle 
\begin{align*}
  \label{eq:def-as}
  \as = \arcsin({d}/{2}),
\end{align*}
and observe that for $(\alpha, \beta) = (\as, 2\pi-\as)$ as well as
for $(\alpha, \beta) = (\pi-\as, \pi+\as)$ we have~$R_S = R_F$
(Figure~\ref{fig:alpha-star}). 
\begin{figure}
  \centerline{\subfigure[$\alpha=\as$, $\beta=2\pi-\as$]{%
      \includegraphics{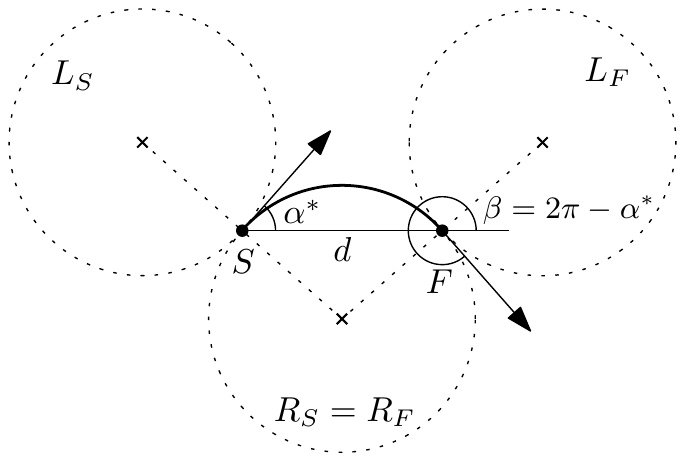}}
    \hspace{1cm}	
    \subfigure[$\alpha=\pi-\as$, $\beta=\pi+\as$]{%
      \includegraphics{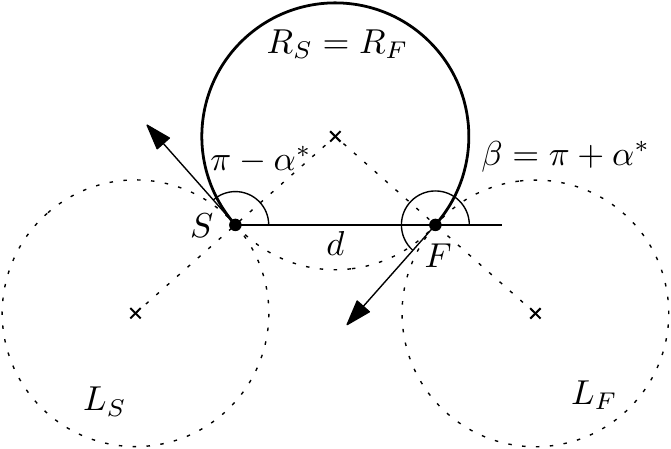}}}
  \caption{Configurations when $R_S$ and $R_F$ coincide.}
  \label{fig:alpha-star}
\end{figure}
Let us also define $\sigmas=\arcsin({d}/4)$.

The following lemma is proven by elementary calculations, given in the
appendix.
\begin{lemma}
  \label{lem:region-boundaries}
  For $0 < d < 2$,
  \begin{denseitems}
  \item there is a curve $(\sigma, \deltalr(\sigma))_{0 \leq \sigma \leq
    \pi}$ in $\SQUARE$ that connects the two points $(0,\as)$ and
    $(\pi,\as)$, lies strictly between $\delta = \as$ and $\delta = \pi/2$
    except for its endpoints, and such that $\dlr = 2$ on the curve, $\dlr
    < 2$ between the curve and the line $\delta = \pi/2$, and $\dlr > 2$
    below the curve;
  \item there is a curve $(\sigma, \deltarl(\sigma))_{0 \leq \sigma \leq
    \sigmas}$ in $\SQUARE$ that connects the two points $(0,\as)$ and
    $(\sigmas, 0)$, lies strictly below $\delta = \as$ except for its left
    endpoint, and such that $\drl = 2$ on the curve, $\drl < 2$ between the
    curve and the line $\delta=\as$, and $\drl > 2$ below the curve;
  \item for $\as \leq \delta \leq \pi/2$, we have $\drl \leq 2$ with equality
    only for the two points $(0,\as)$ and~$(\pi, \as)$;
  \item for $\sigmas < \sigma < \pi-\sigmas$, $0 \leq \delta \leq \as$, we
    have~$\drl < 2$. 
\end{denseitems}
\end{lemma}

By Equations~\eqref{eq:dlr-sigma-delta} and~\eqref{eq:drl-sigma-delta},
we have $\drl(\sigma,\delta) = \dlr(\sigma,\pi-\delta)$.  Our regions
are therefore as follows (see Figure~\ref{fig:regions-1}):
\begin{denseitems}
\item $C_1$ is the region $\delta \leq \deltalr(\sigma)$. Inside this
  region we have $\dlr\geq 2$.
\item $C_2$ is the region $\pi-\sigmas
  \leq \sigma \leq \pi$, $\pi - \deltarl(\sigma) \leq \delta \leq
  \pi$.  Here we have $\dlr \geq 2$ and $\drl \geq 2$.
\item $C_3$ is the region $0 \leq \sigma \leq \sigmas$, $\pi -
  \deltarl(\sigma) \leq \delta \leq \pi$. Here we have $\dlr \geq 2$
  and $\drl \geq 2$.
\item $A$ is the region $\deltalr(\sigma) < \delta <
  \pi-\deltalr(\sigma)$. In this region we have $\dlr < 2$ and $\drl <
  2$.
\item Finally, $B$~is the remaining region, where
  $\pi-\deltalr(\sigma) \leq \delta$, but excluding~$C_2 \cup C_3$.
  In this region we have $\dlr < 2$ and $\drl \geq 2$ .
\end{denseitems}

\begin{figure}
  \centerline{\subfigure[Subdivision of $\SQUARE$]{%
      \includegraphics{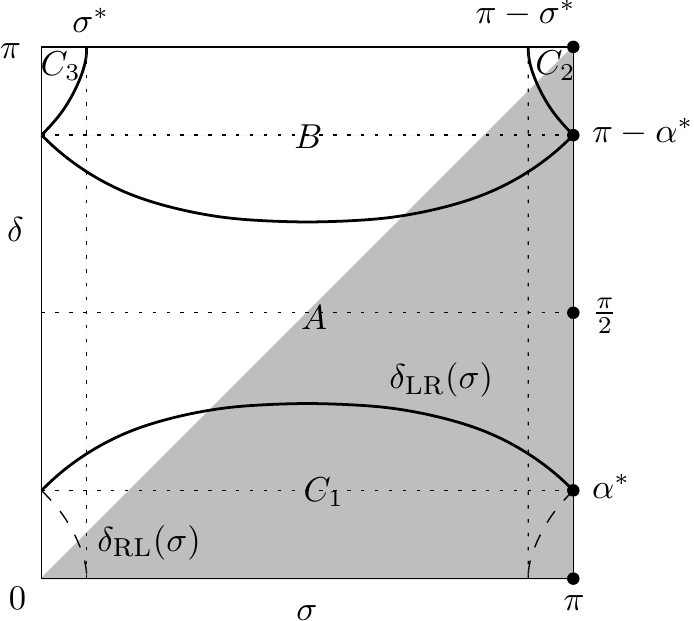}%
      \label{fig:regions-1-sigma-delta}}
    \hspace{2cm}
    \subfigure[Subdivision of $\Delta$]{%
      \includegraphics{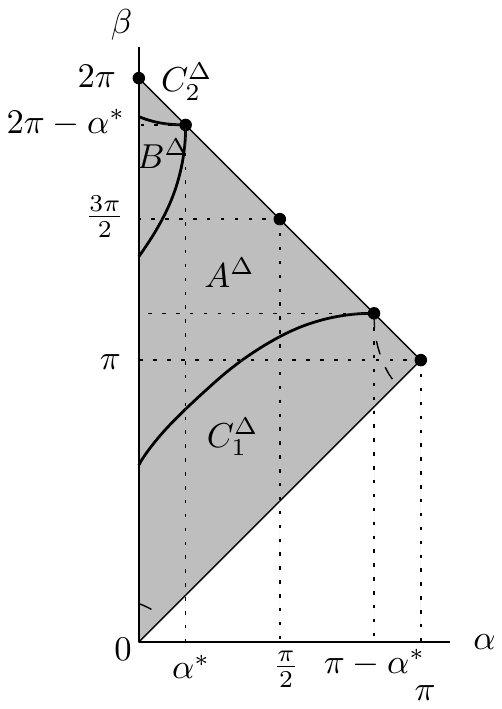}%
      \label{fig:regions-1-alpha-beta}}}
  \caption{The subdivision of $\SQUARE$ and $\Delta$ for $d=1$}
  \label{fig:regions-1}
\end{figure}

It is clear from this description that the five regions~$A$, $B$,
$C_1$, $C_2$, and $C_3$ are $\sigma$-monotone, meaning that a line
parallel to the $\delta$-axis intersects each region in a single
interval.  We will also need that the region~$\BD$ is monotone with
respect to the $\alpha$-direction.  The proof by elementary
calculations is again given in the appendix.
\begin{lemma}
  \label{lem:case-b-alpha-monotone}
  For $0 < d < 2$, there are two continuous functions $\alpha \mapsto
  \brl(\alpha)$ and $\alpha \mapsto \blr(\alpha)$ defined on the
  interval $[0, \as]$ such that $\blr(\as) = \brl(\as) = 2\pi -\as$,
  and such that for $0 \leq \alpha < \as$ we have
  \begin{denseitems}
  \item $\alpha + \pi < \brl(\alpha) < 2\pi - \as < \blr(\alpha) <
    2\pi-\alpha$;
  \item $\drl < 2$ for $\beta < \brl(\alpha)$,
    $\drl = 2$ for $\beta = \brl(\alpha)$, and
    $\drl > 2$ for $\beta > \brl(\alpha)$;
  \item $\dlr < 2$ for $\beta < \blr(\alpha)$,
    $\dlr = 2$ for $\beta = \blr(\alpha)$, and
    $\dlr > 2$ for $\beta > \blr(\alpha)$;
  \end{denseitems}
  The function $\blr$ is a monotonically decreasing function
  of~$\alpha$, and we have 
  \[
  \BD = \big\{(\alpha, \beta) \mathrel{\big|} \alpha\in [0,\as),\; 
    \brl(\alpha) \leq \beta < \blr(\alpha) \big\}.
  \]
\end{lemma}

\section[Explicit expressions for the length of LRL- and RLR-paths]%
        {Explicit expressions for the length of \plrl- and \prlr-paths}
\label{sec:ccc-paths}

In this section we develop explicit formulas for the length of \plrl-
and \prlr-paths.

We start by a change of perspective, and consider all configurations
where~$\dl$ is fixed.  We choose a coordinate system where the line
$\ell_S\ell_F$ is horizontal, and $\ell_S$ lies to the left
of~$\ell_F$, see Figure~\ref{fig:lrl-regions}.
\begin{figure}
  \centerline{\subfigure[\plrl-paths]{\includegraphics{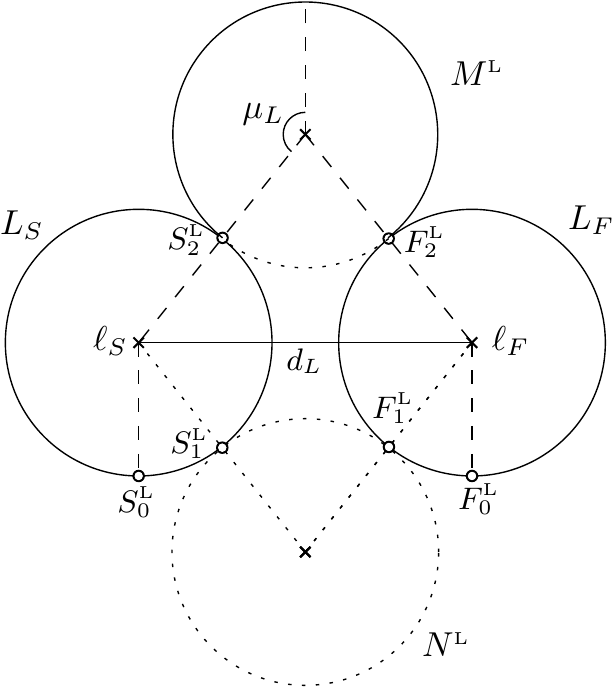}%
      \label{fig:lrl-regions}}
    \hspace{2cm}
    \subfigure[\prlr-paths]{\includegraphics{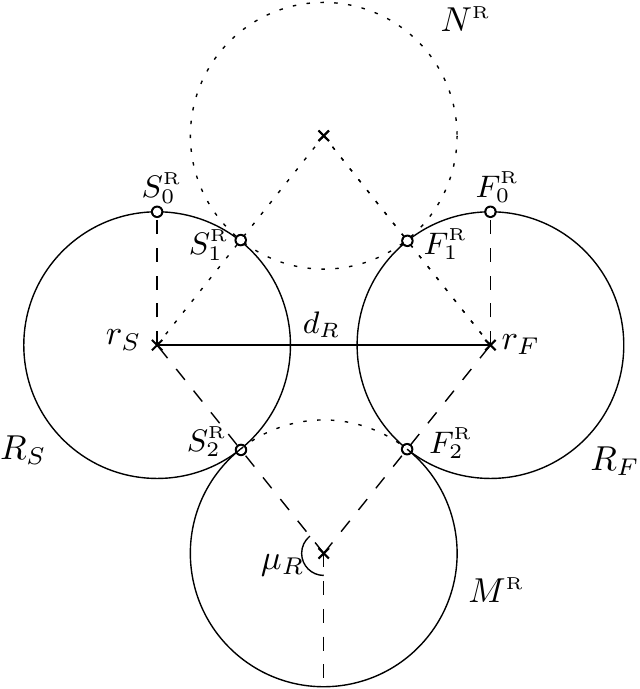}%
      \label{fig:rlr-regions}}}
  \caption{Locations of $S$ and $F$ on their disks}
  \label{fig:lrl-rlr-regions}
\end{figure}
We have drawn the two unit-radius disks~$M^{\tL}$ and~$N^{\tL}$
tangent to $L_S$ and~$L_F$.  The points of tangency are $\SL_2$ and
$\FL_2$ for~$M^{\tL}$, and $\SL_1$ and $\FL_1$ for~$N^{\tL}$.
Dubins~\cite{d-cmlcvcpitpt-57} showed that the length of the middle
circular arc of a \pccc-path is larger than $\pi$, and so it lies
on~$M^{\tL}$.

So any \plrl-path first follows a leftwards arc on~$L_S$, then
switches to~$M^{\tL}$ at~$\SL_2$, follows the rightwards arc
on~$M^{\tL}$ until it reaches~$\FL_2$, and finally follows a leftwards
arc on~$L_F$.  We note that the middle arc on~$M^{\tL}$ does not
depend on the specific endpoints~$S$ and~$F$, it is determined
entirely by $\SL_2$ and~$\FL_2$, and therefore by~$\dl$.  Let $\mul$
denote half the length of the middle circular arc~$\arc{\SL_2\FL_2}$.  We
have $\pi/2 < \mul \leq \pi$, and $4\sin(\pi-\mul) = \dl$, so that we
have
\begin{align*}
  \mul &= \pi - 4\arcsin ({\dl}/{4}).
\end{align*}
The same considerations apply to \prlr-paths, see
Figure~\ref{fig:rlr-regions}. We define $\mur$ as half the length of
the middle circular arc of the \prlr-path, and obtain
\begin{align*}
  \mur &= \pi - 4\arcsin ({\dr}/{4}).
\end{align*}
It is easy to express the length of \pccc-paths up to a multiple
of~$2\pi$:
\begin{lemma}
  \label{lem:lrl-length-general}
  For any $\sigma, \delta$, we have
  \begin{align*}
    \llrl(\sigma, \delta) & \equiv 4\mul + 2\delta \pmod{2\pi},\\
    \lrlr(\sigma, \delta) & \equiv 4\mur - 2\delta \pmod{2\pi}.
  \end{align*}
\end{lemma}
\begin{proof}
  An \plrl-path consists of an initial left-turning arc of
  length~$\gamma_1$ on $L_S$, a right-turning arc of length $2\mul$ on
  the middle disk, and a final left-turning arc of length~$\gamma_{2}$
  on~$L_F$.  This means that the total change in orientation is
  $\gamma_1 - 2\mul + \gamma_2$.  On the other hand, since the initial
  orientation is~$\alpha$ and the final orientation is~$\beta$, this
  must be equal, up to multiples of~$2\pi$, to $\beta - \alpha =
  2\delta$.  It follows that
  \[
  \llrl = \gamma_1 + \gamma_2 + 2\mul = 
  \gamma_1 - 2\mul + \gamma_2 + 4\mul \equiv 2\delta + 4\mul \pmod{2\pi}.
  \]
  For \prlr-paths, we can similarly observe that $-\gamma_1 + 2\mur -
  \gamma_2 \equiv 2\delta \pmod{2\pi}$ (here, $\gamma_1$
  and~$\gamma_{2}$ are the right-turning arcs) and obtain
  \[
  \lrlr = \gamma_1 + \gamma_2 + 2\mur = 
  4\mur - (-\gamma_1 + 2\mur - \gamma_2) \equiv 4\mur - 2\delta
  \pmod{2\pi}. 
  \qedhere
  \]
\end{proof}
Lemma~\ref{lem:lrl-length-general} expresses the length of \pccc-paths
only up to a multiple of~$2\pi$.  Indeed, when considering the length
as a function of the endpoints~$S$ and~$F$ on the fixed disks~$L_S$
and~$L_F$, then the length jumps by~$2\pi$ when~$S$ crosses~$S_2^\lL$,
or when~$F$ crosses~$F_2^\lL$ (see Figure~\ref{fig:lrl-regions}), even
though~$\mul$ is constant and~$\delta$ changes continuously.  It is
therefore important to understand the possible locations of the
endpoints~$S$ and~$F$ along the disks~$L_S$ and~$L_F$:
\begin{lemma}
  \label{lem:ccc-endpoint-locations}
  We have 
  \begin{denseitems}
  \item $S$ lies on the counter-clockwise arc $\arc{\SL_1\SL_2}$ 
    of~$L_S$ if and only if $\drl < 2$;
  \item $F$ lies on the clockwise arc~$\arc{\FR_2\FR_1}$ of~$R_F$
    if and only if $\drl < 2$;
  \item $S$ lies on the clockwise arc $\arc{\SR_1\SR_2}$ of~$R_S$
    if and only if $\dlr < 2$;
  \item $F$ lies on the counter-clockwise arc~$\arc{\FL_2\FL_1}$
    of~$L_F$
    if and only if $\dlr < 2$.
  \end{denseitems}
  If $\drl \geq 2$ and in addition $0 \leq \alpha \leq \pi/2$, then $S$
  lies on the counter-clockwise arc $\arc{\SL_0\SL_1}$ of~$L_S$.
\end{lemma}
\begin{proof}
 Consider the position of~$R_S$ as $S$ moves once around the fixed
 circle~$L_S$. The center~$r_S$ describes a circle of radius two
 around~$\ell_S$.  When $S = \SL_1$, we have $R_S = N^{\tL}$, when $S
 = \SL_2$, we have $R_S = M^{\tL}$.  If $\drl < 2$, we have $R_S \cap
 L_F \neq \emptyset$, and so $S$ must lie on the counter-clockwise
 arc~$\arc{\SL_1\SL_2}$.  If $\drl \geq 2$, we have $R_S \cap L_F =
 \emptyset$, and $S$ must lie on the complementary
 arc~$\arc{\SL_2\SL_1}$.  The same argument applies to \prlr-paths to
 determine the location of~$F$.  We argue similarly for 
 $\dlr < 2$ and $\dlr \geq 2$.

 We observe next that the starting orientation at~$S$, which is the
 forward tangent to $L_S$ at~$S$, and the vector
 $\overrightarrow{SF}$ must make an angle of~$\alpha$.  Since $F \in
 L_F$, this is impossible for $0 \leq \alpha \leq \pi/2$ and $S$ on
 the long counter-clockwise arc~$\arc{\SL_2\SL_0}$, and so $\drl \geq
 2$ with $0 \leq \alpha \leq \pi/2$ implies that $S \in
 \arc{\SL_0\SL_1}$.
\end{proof}

Given the location of the endpoints on the disks~$L_S$ and~$L_F$, we
can strengthen Lemma~\ref{lem:lrl-length-general} by computing the
exact multiple of~$2\pi$ that appears in the length.  The proof, based
on bounding the possible lengths of~$\gamma_1$ and~$\gamma_2$, the
initial arc and the final arc of a \pccc-path, is given in the
appendix.
\begin{lemma} 
  \label{lem:lrl-length-specific}
  We have
  \begin{denseitems}
  \item
    For $(\sigma,\delta) \in A$, 
    $\llrl(\sigma, \delta) = 4\mul + 2\delta - 2\pi$ and 
    $\lrlr(\sigma, \delta)  = 4\mur - 2\delta$;
  \item
    For $(\sigma, \delta) \in B$,
    $\llrl(\sigma, \delta) = 4\mul + 2\delta - 2\pi$ and
    $\lrlr(\sigma, \delta) = 4\mur - 2\delta + 2\pi$.
  \end{denseitems}
\end{lemma}

Understanding the position of the endpoints also allows us to extend
Lemma~\ref{lem:monotonicity-csc} to include the \plrl-path, at least
when the original configuration is in
case~$\BD$.  We need this to prove that $\dub(d) = 2\pi$ for all $d \geq
2$ by simply appealing to monotonicity.
\begin{lemma}
  \label{lem:monotonicity-all-but-rlr}
  Let $d_1 < d_2$, and $(\alpha, \beta) \in \BD$ (where $\BD$ is
  defined for~$d_1$).  If there is an \plrl-path of length~$\ell$
  from~$(0, 0, \alpha)$ to~$(d_1, 0, \beta)$, then there is a path of
  length~$\ell + (d_2 - d_1)$ from~$(0, 0, \alpha)$ to~$(d_2, 0,
  \beta)$.
\end{lemma}
\begin{proof}
  We only have to prove that the \plrl-path has a horizontal tangent
  oriented in the positive $x$-direction.  Assume this is not the
  case, so there is no point on the path were the orientation is~$0$
  or~$2\pi$.  The path starts at orientation~$\alpha$, the orientation
  increases to~$\alpha + \gamma_1$, decreases to~$\alpha + \gamma_1 -
  2\mul$, and increases again to~$\beta = \alpha + \gamma_1 - 2\mul +
  \gamma_2$, without ever leaving the open range~$(0, 2\pi)$.  But by
  Lemma~\ref{lem:case-b-basics}, $(\alpha, \beta) \in \BD$ implies
  $\drl \geq 2$ and $0 \leq \alpha \leq \pi/2$ and thus, by
  Lemma~\ref{lem:ccc-endpoint-locations} the point~$S$ lies on the
  counter-clockwise arc~$\arc{\SL_0\SL_1}$ of~$L_S$.  Since $\angle
  \SL_0\l_S\SL_2 = \mul$ and $\angle \SL_1\l_S\SL_2 = 2\pi -\mul$, we
  have $2\mul - \pi \leq \gamma_1 \leq \mul$ and $0 \leq \gamma_2 \leq
  2\mul-\pi$.  This implies that $\gamma_1 + \gamma_2 \leq 2\mul$, and
  so $\beta \leq \alpha$, a contradiction to~$(\alpha, \beta)\in \BD$.
\end{proof}

The explicit expressions for the \plrl-path length in
Lemma~\ref{lem:lrl-length-specific} allow us to study its
derivative, and to show that the length of \plrl-paths is monotone
in~$\alpha$ and~$\beta$, at least for the cases of interest to us.
The calculations involving these derivatives are given in the appendix.
\begin{lemma}
  \label{lem:lrl-changes-alpha-beta}
  We have
  \begin{denseitems}
  \item  For $(\alpha, \beta) \in \AD$, the function $\beta \mapsto
    \llrl(\alpha, \beta)$ is increasing;
  \item 
    For $(\alpha, \beta)\in \BD$, the function $\alpha \mapsto
    \llrl(\alpha, \beta)$ is decreasing, while $\beta \mapsto
    \llrl(\alpha, \beta)$ is increasing; Moreover, if $(\alpha, \beta)
    \in \BD$ and $\beta \geq {3\pi}/{2}$, then
    $\frac{\partial\llrl}{\partial\beta}(\alpha, \beta) \geq 1$.
  \end{denseitems}
\end{lemma}

\section{Case A}
\label{sec:case-a}

When $d < 2$, both \plrl- and \prlr-paths exist for any $(\sigma,
\delta) \in \SQUARE$.  We define three functions $\lL$,~$\lR$,
and~$\lC$ on~$\SQUARE$:
\begin{align}
  \label{eq:def-l}
  \lL(\sigma, \delta) &= 4\mul(\sigma,\delta) + 2\delta -
  2\pi\\ 
  \label{eq:def-r}
  \lR(\sigma, \delta) &= 4\mur(\sigma,\delta) - 2\delta\\
  \label{eq:def-c}
  \lC(\sigma, \delta) &= \min \{ \lL(\sigma, \delta), \lR(\sigma, \delta)\}.
\end{align} 
While these functions are defined and continuous everywhere
on~$\SQUARE$, we have shown in Lemma~\ref{lem:lrl-length-specific}
only that $\llrl(\sigma, \delta) = \lL(\sigma, \delta)$ for
$(\sigma,\delta) \in A \cup B$, and $\lrlr(\sigma, \delta) =
\lR(\sigma, \delta)$ for $(\sigma, \delta) \in A$.  

Recall that $\as = \arcsin({d}/{2})$. We define the following
rectangle $\RECTII \subset \SQUARE$:
\[
\RECTII: \quad 0\leq \sigma \leq \pi \quad \text{and} \quad \as \leq \delta
\leq \pi - \as.
\]
Note that $A \subset \RECTII$ (see
Figure~\ref{fig:regions-1-sigma-delta}).  It will be easier to work
with the rectangular domain~$\RECTII$ rather than the curved
region~$A$, as long as we keep in mind that $\lC(\sigma, \delta)$ is
the length of the shortest \pccc-path for $(\sigma, \delta) \in A$
only.

By studying the derivatives of the functions~$\lL(\sigma,\delta)$
and~$\lR(\sigma, \delta)$, we can show that they are monotone in
$\sigma$- and $\delta$-direction.  The calculations are given in the
appendix. 
\begin{lemma}
  \label{lem:l-r-fixed}
  For $(\sigma, \delta) \in \RECTII$, the function
  \begin{denseitems}
  \item $\sigma \mapsto \lL(\sigma,\delta)$ is decreasing, 
  	while $\sigma \mapsto \lR(\sigma,\delta)$ is increasing,
  \item $\delta \mapsto \lL(\sigma,\delta)$ is increasing, 
    while $\delta \mapsto \lR(\sigma,\delta)$ is decreasing.
  \end{denseitems}
\end{lemma}
Monotonicity of $\lL$ and~$\lR$ implies that neither has a local
extremum in the interior of~$\RECTII$, and a local extremum can only
occur on the set~$\LER$ of points $(\sigma,\delta)$ with
$\lL(\sigma,\delta) = \lR(\sigma,\delta)$ (see
Figure~\ref{fig:regions-1-rect}).  
\begin{figure}
  \centerline{\includegraphics{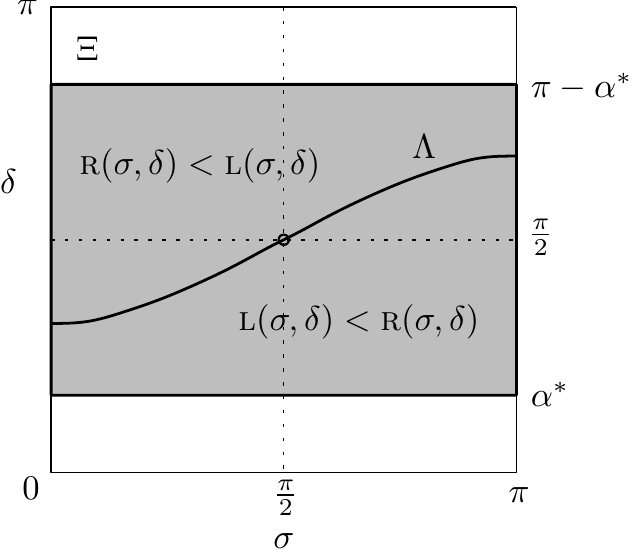}}
  \caption{The curve $\LER$ where $\lL(\sigma, \delta) = \lR(\sigma,
    \delta)$ in rectangle~$\RECTII$ (shaded region), shown for $d =
    1$.}
  \label{fig:regions-1-rect}
\end{figure}
In the appendix we use Lagrange multipliers to prove that in fact the
only local extremum is at~$(\pi/2, \pi/2)$.
\begin{lemma}
  \label{lem:lrl-eq-rlr}
  The function $\lC(\sigma,\delta)$ has no local extremum in the
  interior of~$\RECTII$ except at~$(\pi/2, \pi/2)$.
\end{lemma}

Since $\lC(\sigma,\delta)$ is continuous, it assumes its maximum
on~$\RECTII$. By Lemma~\ref{lem:lrl-eq-rlr}, this must happen either
at $(\sigma,\delta) = (\pi/2, \pi/2)$, or on the boundary
of~$\RECTII$, at a point $(\sigmaa, \deltaa)$ where
$\lL(\sigmaa,\deltaa) = \lR(\sigmaa, \deltaa)$. This must happen
either on the vertical side $\sigma = \pi, \pi/2 < \delta \leq \pi -
\as$, or on the horizontal side $\delta = \pi - \as, \pi/2 < \sigma
\leq \pi$.  The point $(\sigmaa, \deltaa)$ is unique if we require
$\sigmaa + \deltaa > \pi$, and there is a symmetric point $(\pi -
\sigmaa, \pi - \deltaa)$.  Let us define the function $\lA(d)$ for $0
\leq d < 2$ as
\begin{align}
  \label{eq:def-a}
  \lA(d) & = \llrl(\sigmaa, \deltaa) = \lrlr(\sigmaa, \deltaa).
\end{align}
There is an important breakpoint at $d = \sqrt{2}$:
\begin{lemma}
  \label{lem:max-rectii}
  The maximum $\lA(d)$ occurs with $\sigmaa = \pi$ when $0 \leq d \leq
  \sqrt{2}$, and with $\deltaa = \pi-\as$ when $\sqrt{2} \leq d < 2$.
\end{lemma}
\begin{proof}
  We evaluate $\lL(\pi, \pi-\as)$ and $\lR(\pi, \pi-\as)$. Using
  (\ref{eq:dl-sigma-delta}) and~(\ref{eq:dr-sigma-delta}), we have
  \begin{align*}
    \dl(\pi, \pi-\as) &= 2d \\ 
    \dr(\pi, \pi-\as) &= 0, \\
    \intertext{which implies by (\ref{eq:def-l}) and (\ref{eq:def-r}):}
    \lL(\pi, \pi-\as) &= 4\pi-6\as,\\
    \lR(\pi, \pi-\as) &= 2\pi+2\as.
  \end{align*}
  Since $\as = \arcsin(d/2)$, we have $\lR(\pi, \pi-\as) < \lL(\pi,
  \pi-\as)$ for $d < \sqrt{2}$, equality for $d = \sqrt{2}$, and
  $\lR(\pi, \pi-\as) > \lL(\pi, \pi-\as)$ for $d > \sqrt{2}$.  In the
  first case, Lemma~\ref{lem:l-r-fixed} implies that the maximum must
  occur on the vertical side at $\sigma = \pi$, in the last case it
  must occur on the horizontal side at $\delta = \pi - \as$. For $d =
  \sqrt{2}$ the maximum occurs at the corner $(\sigma_A,\delta_A)=(\pi,
  \pi-\as)$.
\end{proof}

Again using Lagrange multipliers, we prove in the appendix that the
functions $d \mapsto \lA(d)$ and $d\mapsto \lA(d) -d $ are monotone.
\begin{lemma}
  \label{lem:a-monotonicity}
  On the interval $0 \leq d \leq \sqrt{2}$, the function
  \begin{denseitems}
  \item $d \mapsto \lA(d)$ is monotonically increasing,
  \item $d \mapsto \lA(d) - d$ is monotonically decreasing.
  \end{denseitems}
\end{lemma}

It remains to decide whether $\lC(d, \pi/2,\pi/2)$ or $\lA(d)$ is larger.
\begin{lemma}
  \label{lem:region1-dist-less-sqrt2}
  For $0 < d \leq \sqrt{2}$, we have
  $\max_{(\sigma,
    \delta)\in\RECTII} \lC(d, \sigma, \delta) = \lA(d)$.
\end{lemma}
\begin{proof}
  For $0 \leq d \leq \sqrt{2}$ and $(\sigma, \delta) = (\pi/2,
  \pi/2)$, we have
  \begin{align*}
    \dl^{2} &= d^{2} + 4 \quad \text{by~(\ref{eq:dl-sigma-delta})}\\
    \mul &= \pi - \arcsin(\sqrt{d^{2} + 4} / 4)\\
    \lL(d, \pi/2, \pi/2) &= 4\mul + 2\delta - 2\pi = 3\pi - 4\arcsin
    (\sqrt{d^{2} + 4} / 4).
  \end{align*}
  Since $\arcsin$ is an increasing function, $d \mapsto \lL(d, \pi/2,
  \pi/2)$ is a decreasing function. We therefore have
  \[
  \lL(d, \pi/2, \pi/2) < \lL(0, \pi/2, \pi/2) = 7\pi/3.
  \]
  On the other hand, by Lemma~\ref{lem:a-monotonicity}, the function
  $d \mapsto \lA(d)$ is increasing, and so $\lA(d) \geq \lA(0) =
  7\pi/3 > \lL(d, \pi/2, \pi/2)$.
\end{proof}

We now justify that it suffices to study \pccc-paths in case~A, as no
other path type can be shorter.  Since \plsr- and \prsl-paths do not
exist, it is enough to show the following lemma:
\begin{lemma}
  \label{lem:region2-lrl-rlr-shorter}
  For~$(\sigma, \delta) \in A$, we have 
  \begin{align*}
    \llrl(\sigma, \delta) & \leq \llsl(\sigma, \delta) \\
    \lrlr(\sigma, \delta) & \leq \lrsr(\sigma, \delta).
  \end{align*}
\end{lemma}
\begin{proof}
  Let $\gamma_1$ and $\gamma_2$ be the length of the left-turning arcs
  of an \plrl-path.  By Lemma~\ref{lem:ccc-endpoint-locations}, the
  endpoints~$S$ and~$F$ lie on the counterclockwise arcs
  $\arc{\SL_1\SL_2}$ of~$L_S$ and $\arc{\FL_2\FL_1}$ of~$L_F$ (see
  Figure~\ref{fig:lrl-regions}).

  On the other hand, the \plsl-path turns left on~$L_S$ until~$\SL_0$,
  goes along the tangent to~$\FL_0$, then turns left on~$L_F$ until it
  reaches~$F$.  Since $\angle \SL_2\ell_S\SL_0 = \mul$ and
  $|\SL_0\FL_0| = \dl$, we have
  \begin{align*}
    \llsl - \llrl & = \dl + 2(2\pi - \mul) - 2\mul = \dl + 4(\pi -
    \mul) \geq 0
  \end{align*}
  since $\mul \leq \pi$. The analogous argument shows that $\lrlr \leq \lrsr$.
\end{proof}

Putting everything together, the following two lemmas
describe~$\dub_A(d)$.
\begin{lemma}
  \label{lem:a-less-sqrt2}
  For $0 < d < \sqrt{2}$, $\dub_{A}(d) = \lA(d) - d$.  In other words,
  the maximum is realized by the unique point~$(\sigmaa, \deltaa)$ on
  the segment $\sigma = \pi$, $\pi/2 \leq \delta \leq \pi - \as$ where
  $\lL(d, \sigmaa,\deltaa) = \lR(d, \sigmaa, \deltaa)$.  We have
  $\lA(d) > 2\pi + 2\as$.
\end{lemma}
\begin{proof}
  By Lemma~\ref{lem:region1-dist-less-sqrt2}, we have
  $\max_{(\sigma,\delta) \in \RECTII} \lC(d, \sigma, \delta) =
  \lA(d)$, and the maximum is assumed at the point $(\sigmaa,
  \deltaa)\in \AD$.  By Lemma~\ref{lem:region2-lrl-rlr-shorter}, we
  have $\l(d, \sigma, \delta) = \lC(d, \sigma,\delta)$ for $(\sigma,
  \delta) \in \AD$.  Since $(\sigmaa, \deltaa) \in \AD \subset
  \RECTII$, this means that $\dub_{A}(d) = \sup_{(\sigma,\delta) \in
    \AD}\l(d, \sigma, \delta)-d = \lC(d, \sigmaa, \deltaa)-d =
  \lA(d)-d$.  Finally, the last inequality follows from $\lA(d) >
  \lC(\pi, \pi-\as) = \lR(\pi,\pi -\as) = 2\pi+ 2\as$, as we
  observed in the proof of Lemma~\ref{lem:max-rectii}.
\end{proof}

\begin{lemma}
  \label{lem:a-larger-sqrt2}
  For $\sqrt{2} \leq d < 2$, we have $\dub_{A}(d) \leq \max\{2\pi,
  \dub_{B}(d)\}$.
\end{lemma}
\begin{proof}
  Let $\bar{A}$ be the closure of~$\AD$. Since $\bar{A}$ is compact
  and $\lC$ is continuous, there is a point $(\sigma, \delta) \in
  \bar{A}$ where $\lC(\sigma,\delta)$ assumes its maximum.  By
  Lemma~\ref{lem:lrl-eq-rlr} this is necessarily a point where
  $\lL(\sigma, \delta) = \lR(\sigma, \delta)$, and either
  $(\sigma,\delta) = (\pi/2, \pi/2)$, or $(\sigma,\delta)$ lies on the
  boundary of~$\bar{A}$.  By Lemma~\ref{lem:max-rectii}, it cannot
  occur on the vertical side of~$\RECTII$.
  
  Assume first that $\delta < \pi/2$. By Lemma~\ref{lem:l-r-fixed}, we
  must then have $\sigma < \pi/2$. Using Lemmas~\ref{lem:l-r-fixed}
  and~\ref{lem:lrl-changes-alpha-beta}, we have
  \begin{align*}
    \lC(d, \sigma, \delta) & \leq \lL(d, \sigma, \delta) \leq \lL(d, \delta,
    \delta) \leq \lL(d, \pi/2, \pi/2).
  \end{align*}
  We observed in the proof of Lemma~\ref{lem:region1-dist-less-sqrt2} that
  $\lL(d, \pi/2, \pi/2)$ is a decreasing function of~$d$.  For $d = \sqrt{2}$,
  we already have $\lL(\sqrt{2}, \pi/2, \pi/2) = 3\pi - 4 \arcsin(\sqrt{6}/4) <
  2\pi + \sqrt{2}$, and so $\lC(d, \sigma, \delta) \leq d + 2\pi$.  The same
  argument covers the case where $(\sigma, \delta) = (\pi/2, \pi/2)$.

  It remains to consider the possibility that $\delta > \pi/2$.  In
  this case $(\sigma, \delta)$ lies on the common boundary of $\AD$
  and~$\BD$.  On this boundary we have $\drl = 2$, so $R_S$ and $L_F$
  touch. Note that in this case $(\sigma,\delta)$ lies in $\BD$, not
  in $\AD$.  It remains to observe that then the \prlr-path is
  identical to the \prsl-path, so we have $\lrsl(\sigma, \delta) =
  \llrl(\sigma, \delta)$.  We will prove in
  Lemma~\ref{lem:region3-lrl-rsl-shorter} that in $\BD$ these two path
  types are always shortest, and so $\lC(d, \sigma, \delta) \leq
  \dub_{B}(d) + d$.
\end{proof}


\section{Case B}
\label{sec:case-b}

It was proven by Goaoc et al.~\cite{gkl-bcsp-2010} that for any
\pcsc-path type (that is, one of the types \plsr, \prsl, \plsl, or
\prsr), the length of a path of this type from $(0, 0, \alpha)$ to
$(d, 0, \beta)$ is differentiable at any point $(\alpha,\beta) \in
\SQUAREAB$ where such a path exists and both its circular arcs have
non-zero length.  For the case of \prsl-paths, they prove specifically
that
\begin{align}
  \label{eq:rsl-dalpha}
  \frac{\partial}{\partial\alpha} \lrsl(\alpha, \beta) & = 
  1 - \cos \gamma_\tR\\
  \label{eq:rsl-dbeta}
  \frac{\partial}{\partial\beta} \lrsl(\alpha, \beta) & = 
  1 - \cos \gamma_{\tL},
\end{align}
where $\gamma_{\tR}$ and~$\gamma_{\tL}$ are the lengths of the
right-turning and the left-turning circular arc on the path.

We recall that case~B is the situation where $\dlr(\alpha, \beta) < 2$
and $\drl(\alpha, \beta) \geq 2$.  For $0 < d < 2$,
Lemma~\ref{lem:case-b-alpha-monotone} gives an explicit description of
the region~$\BD$, using the two functions~$\brl(\alpha)$
and~$\blr(\alpha)$. Let us define two extended regions:
\begin{align*}
  \Bc & = \big\{ (\alpha, \beta) \mathrel{\big|} 0 \leq \alpha \leq \as, \;
  \brl(\alpha) \leq \beta \leq 2\pi - \alpha\big\}, \\
  \Bb & = \big\{ (\alpha, \beta) \mathrel{\big|} 0 \leq \alpha \leq \as, \;
  \brl(\alpha) \leq \beta \leq \blr(\alpha) \big\}.
\end{align*}
So $\Bb$ is the closure of~$\BD$, while $\Bc$ is the union~$\BD \cup
\CD_{2}$ (see Figure~\ref{fig:regions-1-alpha-beta}).

We now investigate where the three segments of an \prsl-path can
vanish in~$\Bc$: First, the $S$-segment vanishes exactly if $\drl =
2$.  This happens exactly on the lower boundary of~$\Bc$.  By
Lemma~\ref{lem:ccc-endpoint-locations}, $S$~lies on the
arc~$\arc{\SL_0\SL_1}$ of~$L_S$ (see
Figure~\ref{fig:lrl-rlr-regions}), and so $0 \leq \gamma_{\tR} \leq
2\as$.  If $\gamma_{\tR} = 0$, then we must have $S = \SL_0$ and
therefore $F = \FL_0$.  This is the case $(\alpha, \beta) = (0,
2\pi)$.  Finally, if $\gamma_{\tL} = 0$, then~$F$ must lie on the
arc~$\arc{\FL_1 \FL_0}$, and we have $2\pi - \beta \leq \alpha$ (see
Figure~\ref{fig:gammal}).  Equality holds only for $S = \SL_1$, $F =
\FL_1$, which is the case $(\alpha, \beta) = (\as, 2\pi-\as)$.  In all
other cases, $2\pi - \beta < \alpha$ is a contradiction to $(\alpha,
\beta) \in \Delta$.
\begin{figure}
  \centerline{\includegraphics{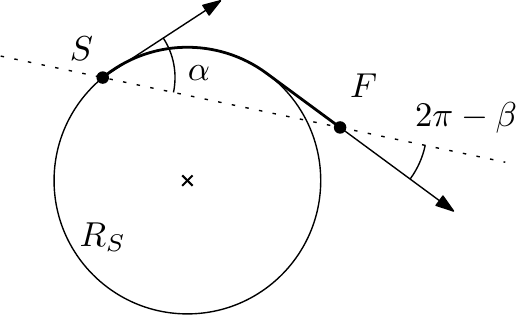}}
  \caption{When $\gamma_{\tL}$ vanishes.}
  \label{fig:gammal}
\end{figure}

The above implies that $\lrsl(\alpha, \beta)$ is differentiable in any
point in the interior of~$\Bc$. The function is continuous everywhere
except at the two points $(\alpha, \beta) = (0, 2\pi)$ and
$(\alpha,\beta) = (\as, 2\pi - \as)$.  At the first point, the
\prsl-path degenerates to the line segment~$\overline{SF}$ of
length~$d$, while the limit of $\lrsl(\alpha, \beta)$ for $(\alpha,
\beta) \rightarrow (0, 2\pi)$ is~$d + 2\pi$.  For the second point
$(\as, 2\pi - \as)$, consider Figure~\ref{fig:alpha-star}(a).  At this
point, both the straight segment and the left-turning arc vanish at
the same time, and the length of the path is only~$|\arc{SF}| = 2\as$.
However, for $(\alpha, \beta) \rightarrow (\as, 2\pi -\as)$, the limit
of $\lrsl(\alpha, \beta)$ is $|\arc{SF}| + 2\pi = 2\as + 2\pi$.  We
observe that this is exactly the value of $\lrlr(\as, 2\pi-\as)$, as
the final right-turning arc of the \prlr-path vanishes.

We therefore define the following function on~$\Bc$:
\begin{align*}
  \lrslp(\alpha, \beta) & = \begin{cases}
    d + 2\pi & \text{for $(\alpha, \beta) = (0, 2\pi)$} \\
    2\as + 2\pi & \text{for $(\alpha, \beta) = (\as, 2\pi-\as)$} \\
    \lrsl(\alpha, \beta) & \text{else}
  \end{cases}
\end{align*}

\paragraph{Overview.} The main goal of this section is to show that
for $\sqrt{2} \leq d < 2$, $\dub_B(d)$ is determined by a unique point
$(\alpha_B, \beta_B)$ on the curve $\blr$ such that $\lrslp(\alpha_B,
\beta_B) = \llrl(\alpha_B, \beta_B)$. To prove this, we first show
that the function~$\beta \mapsto \lrslp(\alpha, \beta)$ is increasing
(Lemma~\ref{lem:rsl-changes-alpha-beta}). Since the function~$\beta
\mapsto \llsl(\alpha, \beta)$ is increasing in~$\BD$ as well
(Lemma~\ref{lem:lrl-changes-alpha-beta}), 
the function $\min\{\lrslp(\alpha, \beta), \llrl(\alpha,
\beta)\}$ must assume its maximum on the curve~$\blr$ in $\Bb$.

Lemma~\ref{lem:region3-lrl-rsl-shorter} shows that $\dub_B(d)$ is
indeed determined by the \prsl-path and the \plrl-path only. Finally,
we will prove that the function $d \mapsto \dub_B(d)$ is monotonically
decreasing.

\medskip
We start by studying the derivatives of $\lrslp(\alpha, \beta)$ to
show monotonicity. The short calculations are in the appendix.
\begin{lemma}
  \label{lem:rsl-changes-alpha-beta}
  For $(\alpha, \beta) \in \Bc$, the function
  \begin{denseitems}
  \item $\alpha \mapsto \lrslp(\alpha,\beta)$ is increasing,
  \item $\beta \mapsto \lrslp(\alpha,\beta)$ is increasing,
  \item $\alpha \mapsto \lrslp(\alpha, 2\pi-\alpha)$ is increasing,
  \item $\lrslp(\alpha, \beta) \leq \lrslp(\as, 2\pi-\as) = 2\as + 2\pi$.
  \end{denseitems}
\end{lemma}
Let us define the following function~$\lB(d, \alpha, \beta)$ on~$\Bc$:
\begin{align}
  \label{eq:def-b}
  \lB(d, \alpha, \beta) & = \min \{ \lrslp(d, \alpha, \beta), \, 
  \llrl(d, \alpha, \beta) \}.
\end{align}
Our goal will be to determine $\lB(d) = \sup_{(\alpha, \beta) \in \BD}
\lB(d, \alpha, \beta)$, and then to show that $\dub_{B}(d) =
\lB(d)-d$.  Since~$\lB(d, \alpha, \beta)$ is continuous, we have
$\lB(d) = \max_{(\alpha, \beta) \in \Bb} \lB(d, \alpha, \beta)$.  By
Lemmas~\ref{lem:lrl-changes-alpha-beta}
and~\ref{lem:rsl-changes-alpha-beta}, we have $\lB(d, \alpha, \beta)
\leq \lB(d, \alpha, \blr(\alpha))$ for any $(\alpha, \beta) \in \Bb$,
and so
\[
\lB(d) = \max_{(\alpha, \beta) \in \Bb} \lB(d, \alpha, \beta) = 
\max_{0 \leq \alpha \leq \as} \lB(d, \alpha, \blr(\alpha)).
\]
We now show that $\dub_B(d) = \lB(d) - d$ by showing that the
\plrl-path or the \prsl-path is shorter than any other path type.
Since there is no \plsr-path in case~B, it suffices to exclude path
types~\plsl, \prlr, and~\prsr.  The details can be found in the
appendix.
\begin{lemma}
  \label{lem:region3-lrl-rsl-shorter}
  For $(\alpha, \beta) \in \BD$ we have
  $\llsl(\alpha, \beta) \geq \llrl(\alpha, \beta),$
  $\lrlr(\alpha, \beta) \geq \llrl(\alpha, \beta),$ and
  $\lrsr(\alpha, \beta) \geq \lrsl(\alpha, \beta)$.
\end{lemma}
It remains to understand the function~$\lB(d)$.  By definition, it is
the maximum of the function 
\[
\alpha \mapsto \lB(d, \alpha, \blr(\alpha)) = 
\min\{\lrslp(d,\alpha,\blr(\alpha)),
\llrl(d,\alpha,\blr(\alpha))\}, 
\]
for $0 \leq \alpha \leq \as$ for fixed~$d$.  We first argue that there
is a unique $\alpha_{\tB} \in [0, \as]$ such that
\[
\lB(d, \alpha, \blr(\alpha)) = 
\begin{cases}
\lrslp(d,\alpha,\blr(\alpha)) & \text{ for } \alpha \leq \alpha_\tB \\
\llrl(d,\alpha,\blr(\alpha)) & \text{ for } \alpha \geq \alpha_\tB 
\end{cases}
\]
This follows directly from the following lemma, whose proof (based on
comparing derivatives) can be found in the appendix.
\begin{lemma}
  \label{lem:lrl-rsl-monotone}
  The function $\alpha \mapsto \llrl(\alpha, \blr(\alpha)) -
  \lrsl(\alpha, \blr(\alpha))$ is monotonically decreasing on~$[0,
    \as]$.
\end{lemma}
By Lemmas~\ref{lem:lrl-changes-alpha-beta}
and~\ref{lem:case-b-alpha-monotone}, the function $\alpha 
\mapsto \llrl(d,\alpha, \blr(\alpha))$ is decreasing.  The maximum of
$\lB(d, \alpha, \blr(\alpha))$ is therefore at $\alpha = 0$, at
$\alpha = \alpha_B$, or is a local maximum of $\lrslp(d, \alpha,
\blr(\alpha))$ with $\alpha < \alpha_B$.  

In the first case $\alpha = 0$ a simple geometric argument similar to
Lemma~\ref{lem:dub_c_upper_bound} shows that $\lrsl(d, 0, \blr(0))
\leq d + 2\pi$.  Similarly, we show that a local maximum of
$\lrslp(d,\alpha,\blr(\alpha))$ implies a path length at most $d +
2\pi$. The proof for this technical detail looks at both the path
geometry and at the derivatives, and can be found in the appendix.
\begin{lemma}
  \label{lem:blr-no-extremum}
  If $\alpha \mapsto \lrslp(d, \alpha, \blr(\alpha))$ has an extremum for
  $0 < \alpha < \as$, then $\lrslp(d, \alpha, \blr(\alpha)) \leq d + 2\pi$.
\end{lemma}
Putting these arguments together, we have shown:
\begin{lemma}
  \label{lem:dub_b}
  For $0 < d < 2$ if $\dub_B(d) > 2\pi$, then $\dub_B(d) = \lB(d,
  \alpha_\tB, \blr(\alpha_\tB))-d$, where $\alpha_\tB$ is the unique
  value in~$[0, \as]$ where $\lrslp(d, \alpha_\tB, \blr(\alpha_\tB)) =
  \llrl(d, \alpha_\tB, \blr(\alpha_\tB))$.
\end{lemma}

It remains to argue the monotonicity of~$d \mapsto \dub_B(d)$.
\begin{lemma}
  \label{lem:dub-b-sqrt2-2}
  For $\sqrt{2} \leq d < 2$, $\dub_B(d)$ decreases monotonically from
  $\dub_B(\sqrt{2}) = 5\pi/2 - \sqrt{2}$ to $\dub_B(\ds) =
  2\pi$ for $\ds \approx 1.5874$.  We have $\dub_B(d) \leq 2\pi$
  for $d \geq \ds$.
\end{lemma}
\begin{proof}
  Assume the function is not monotone, so there is $\sqrt{2} \leq d_1
  < d_4 < 2$ such that $2\pi < \dub_B(d_1) < \dub_B(d_4)$.  Since
  $\dub_B(d)$ is continuous, it assumes its maximum $D := \max_{d \in
    [d_1,d_4]}\dub_B(d)$ on the closed interval~$[d_1,d_4]$. The set
  $\{d \in [d_1, d_4] \mid \dub_B(d) = D\}$ is compact, and so assumes
  its infimum, say at~$d_3$. So we have $\dub_B(d_3) > \dub_B(d_1) >
  2\pi$, and $\dub_B(d) < \dub_B(d_3)$ for $d_1 \leq d < d_3$.

  By Lemma~\ref{lem:dub_b}, we have $\dub_B(d_3) = \lB(d_3,
  \alpha_\tB, \beta_\tB) - d_3$, where $0 < \alpha_\tB \leq
  \arcsin(d_3/2)$ and $\beta_\tB = \blr(\alpha_\tB)$.  Let us define
  $d_2 = \max\{d_1, 2\sin(2\pi - \beta_\tB)\}$. Then $d_1 \leq d_2 <
  d_3$, and $(\alpha_\tB, \beta_\tB) \in \BD$ for $d = d_2$.  Then
  there is an \prsl- or \plrl-path from $(0, 0, \alpha_\tB)$ to $(d_2,
  0, \beta_\tB)$ of length at most~$\dub_B(d_2) + d_2$.  By
  Lemmas~\ref{lem:monotonicity-csc}
  and~\ref{lem:monotonicity-all-but-rlr}, there is then a path of
  length $\dub_B(d_2) + d$ from $(0,0, \alpha_\tB)$ to $(d, 0,
  \beta_\tB)$ for all $d \geq d_2$, implying that $\ell(d, \alpha_\tB,
  \beta_\tB) \leq \dub_B(d_2) + d$, a contradiction to the assumption
  that $\dub_B(d_3) > \dub_B(d_2)$.

  Continuity and monotonicity of $\dub_B(d)$ imply that there must be
  a value~$\ds$ with $\dub_B(\ds) = 2\pi$.  We have numerically
  computed the approximation $\ds \approx 1.5874$.  For a given~$d$,
  we first approximate~$\alpha_\tB$ numerically by binary search on
  the interval~$[0, \as]$ using Lemma~\ref{lem:lrl-rsl-monotone}.  We
  can then compute $\blr(\alpha_\tB)$ and~$\lB(d)$.
\end{proof}

\section*{Acknowledgments}

We thank Sylvain Lazard, Xavier Goaoc, and Jaesoon Ha for helpful
discussions. Janghwan Kim studied the case $d \geq 2$ in this master
thesis at KAIST.

\bibliographystyle{plain}
\bibliography{dubins}

\clearpage
\appendix
\section{Appendix}


\paragraph{Proof of Lemma~\ref{lem:region-boundaries}}

\begin{proof}
  We first argue about the curve~$\deltalr(\sigma)$.
  From Eq.~\eqref{eq:dlr-sigma-delta} we have $\dlr^{2} = d^{2} +
  4d\cos\delta\sin\sigma + 4\cos^{2}\delta$.  For $\delta = \as$, we
  have $\cos^{2}\delta = 1-\sin^{2}\delta = 1 - d^{2}/4$, and so
  $\dlr^{2} = 4 + 4d\cos\as\sin\sigma$.  This is equal to~$4$ for
  $\sigma\in\{0,\pi\}$, and otherwise larger than~$4$.  For $\delta =
  \pi/2$, we have $\dlr^{2} = d^{2} < 4$.  Finally, for $\delta \in
  (0, \pi/2)$, we have $\frac{\partial}{\partial\delta}\dlr^2 =
  -4d\sin\delta\sin\sigma-8\cos\delta\sin\delta < 0$, and so $\delta
  \mapsto \dlr$ is a decreasing function for $\sigma \in (0, \pi)$,
  proving the first claim.

  Consider now $\drl^{2} = d^{2} - 4d\cos\delta\sin\sigma +
  4\cos^{2}\delta$ by Eq.~\eqref{eq:drl-sigma-delta}. For $\as \leq \delta
  \leq \pi/2$, we have $\cos^{2}\delta = 1 - \sin^{2}\delta \leq 1 -
  d^{2}/4$, and so $\drl^{2} \leq 4$, with equality only for $\sigma
  \in \{0,\pi\}$ and $\delta = \as$ which proves the third claim.  On
  the interval $0 \leq \sigma \leq \sigmas$, we have $\drl^{2} \geq 4$
  for $\delta = 0$ (with equality only for $\sigma = \sigmas$),
  $\drl^{2} \leq 4$ for $\delta = \as$ (with equality only for $\sigma
  = 0$), and $\drl^{2} \geq 4$ for $\delta = \pi$.  Since for
  fixed~$\sigma$, $\drl^{2} = 4$ is a quadratic polynomial in
  $\cos\delta$, it has at most two roots in $0 \leq \delta \leq \pi$,
  and thus there must be a unique value $0\leq \deltarl(\sigma) \leq
  \as$ on the interval $0 \leq \sigma \leq \sigmas$ where $\drl =
  2$. This proves the second claim.

  In the remaining region $\sigmas < \sigma < \pi-\sigmas$, $0 \leq
  \delta \leq \as$, we have $\drl^2 < 4$. Indeed, in this region we
  have $\sin\sigma > d/4$, and so $\drl^2 < D(\sigma, \delta)$, where
  $D(\sigma, \delta) = d^2-d^2\cos\delta+4\cos^2\delta$.  We have
  $\frac{\partial}{\partial\delta}D(\sigma, \delta) =
  d^2\sin\delta-8\cos\delta\sin\delta =
  \sin\delta(d^2-8\cos\delta)$. Since $\sin^2\delta \leq d^2/4$, we
  have $\cos^2\delta \geq 1-d^2/4$, which gives $64\cos^2\delta \geq
  64-d^2/16 \geq d^4$ since $d < 2$. This means that $d^2 \leq
  8\cos\delta$, which implies that $\delta \mapsto D(\sigma, \delta)$
  is decreasing. Since $D(\sigma, 0) = 4$, the last claim follows.
\end{proof}

\paragraph{Proof of Lemma~\ref{lem:case-b-alpha-monotone}}
\begin{proof}
  Let us fix an $\alpha\in[0,\as)$, so the points $\ell_S$ and~$r_S$
    are fixed.  While $\beta$ ranges over $[0, 2\pi]$, the point
    $\ell_F$ makes a full circle around~$F$.  This means that the
    distance $\drl$ is strictly increasing for half a period, and
    strictly decreasing for the other half period.  This implies that
    in the range $\alpha + \pi \leq \beta \leq 2\pi -\alpha$ there is
    at most one extremum of~$\drl$.  The same argument shows that
    $\dlr$ has at most one extremum in the range.

    Consider first $\beta = 2\pi - \as$.  We observe from
    Figure~\ref{fig:alpha-star}(a) that $S = (0,0)$ lies on the
    boundary of~$R_F$, and so $\dlr < 2$. Also, since $\alpha < \as$,
    $\drl > 2$.  Consider now $\beta = \alpha + \pi$, so $\delta =
    \pi/2$.  By Eqns.~\eqref{eq:dlr-sigma-delta}
    and~\eqref{eq:drl-sigma-delta}, we have $\dlr = \drl = d < 2$.
    Finally, consider $\beta = 2\pi - \alpha$, so that $\sigma = \pi$.
    Since $\alpha < \as$ we have $\delta > \pi - \as$, so $\sin \delta
    < d/2$, so $\cos^{2} \delta = 1 - \sin^{2}\delta > 1 - d^{2}/4$.
    Again by Eqns.~\eqref{eq:dlr-sigma-delta}
    and~\eqref{eq:drl-sigma-delta}, we then have $\dlr^{2} = \drl^{2}
    = d^{2} + 4\cos^{2}\delta > 4$, and so $\dlr = \drl > 2$.

    Since $\dlr = \drl < 2$ for $\beta = \alpha + \pi$, $\dlr = \drl >
    2$ for $\beta = 2\pi - \alpha$, and both functions have only one
    extremum in this range, both functions must assume the value two
    exactly once in this range, at values $\blr(\alpha)$ and
    $\brl(\alpha)$.  For $\beta = 2\pi -\as$ we have $\dlr < 2$ and
    $\drl > 2$, so we have $\brl(\alpha) < 2\pi - \as < \blr(\alpha)$.
    The two functions are clearly continuous, and since we have $\dlr
    = \drl = 2$ for $(\alpha, \beta) = (\as, 2\pi - \as)$ (see
    Figure~\ref{fig:alpha-star}), we have $\blr(\as) = \brl(\as) =
    2\pi-\as$.
  
    Pick any point $(\alpha, \blr(\alpha))$.  This is a configuration
    where $L_S$ and $R_F$ are touching.  If we now increase $\alpha$,
    the point~$\ell_S$ rotates left around~$S$, and so the distance
    $\dlr$ increases (at least locally).  This implies that
    $\blr(\alpha)$ is a decreasing function of~$\alpha$.

  By Lemma~\ref{lem:case-b-basics}, $\BD$ is contained in the triangle
  $0 \leq \alpha \leq \pi/2$, $\alpha + \pi \leq \beta \leq 2\pi
  -\alpha$.  Consider Figure~\ref{fig:alpha-star}(a).  We first
  observe that for $\as < \alpha \leq \pi/2$, the point~$F = (d, 0)$
  is contained in the interior of~$R_S$, and so $\drl < 2$.  It
  follows that for $(\alpha, \beta) \in \BD$ we must have $\alpha \in
  [0, \as)$, and the expression for~$\BD$ follows from the above.
\end{proof}


\paragraph{Proof of Lemma~\ref{lem:lrl-length-specific}}
\begin{proof}
  We note the following angles (see Figure~\ref{fig:lrl-rlr-regions}):
  \begin{align}
    \label{eq:gamma1}
    \angle \SL_1 \ell_S \SL_2 & =
    \angle \FL_2 \ell_F \FL_1 = 2\mul - \pi,\\
    \label{eq:gamma2}
    \angle \SR_1 r_S \SR_2 & =
    \angle \FR_2 r_F \FR_1 = 2\mur - \pi.
  \end{align}
  Let us first assume that $(\sigma,\delta) \in A$.  We have $2\as <
  2\delta < 2\pi-2\as$. On the other hand, $\gamma_1 + \gamma_2 -
  2\mul \geq -2\mul \geq -2\pi$, where $\gamma_1$ is an initial
  left-turning arc of length on $L_S$ and $\gamma_2$ is a final
  left-turning arc of length on~$L_F$ of an \plrl-path (see
  Figure~\ref{fig:feasible_region2_lrl}).  By
  Lemma~\ref{lem:ccc-endpoint-locations}, $S \in \arc{\SL_1\SL_2}$ and
  $F \in \arc{\FL_2\FL_1}$, and so we can extend the \plrl-path to a
  complete clockwise loop as in Figure~\ref{fig:region2_loop}.  The
  loop uses additional left-turns $\zeta_1$ and $\zeta_2$, and an
  additional right-turn of length~$2\mul$.  The total turning angle of
  a clockwise loop is~$-2\pi$, and thus $\gamma_1 + \gamma_2 +
  \delta_1 + \delta_2 - 4\mul = -2\pi$.  Since $2\mul \leq 2\pi$ this
  implies that $\gamma_1 + \gamma_2 - 2\mul \leq 0$.  From $-2\pi \leq
  \gamma_1 + \gamma_2 - 2\mul \leq 0$ and $0 \leq 2\as< 2\delta
  <2\pi-2\as \leq 2\pi$, we conclude that $\gamma_1 + \gamma_2 -
  2\mul\equiv 2\delta \pmod{2\pi}$ implies $\gamma_1 + \gamma_2 -
  2\mul = 2\delta - 2\pi$.  This shows that $\llrl = 4\mul +2\delta -
  2\pi$.
  \begin{figure}
    \centerline{\subfigure[$0 \leq \gamma_1, \gamma_2 \leq 2\mul-\pi$]{%
	\includegraphics{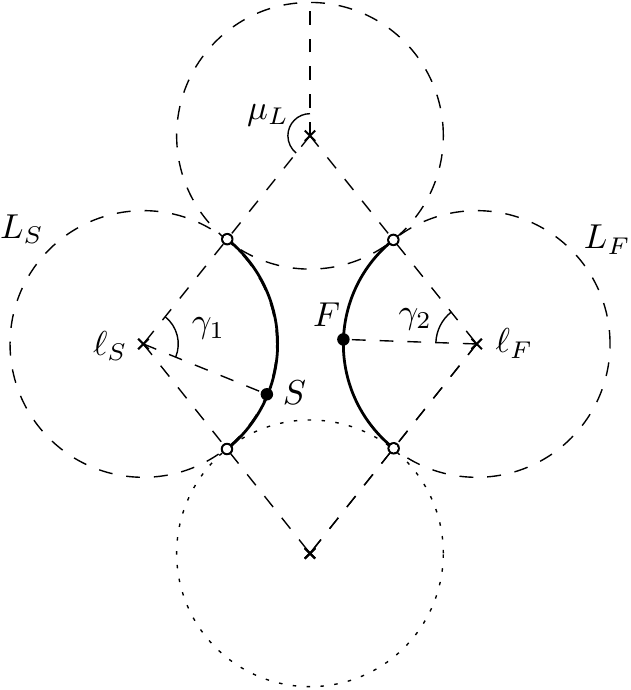}%
	\label{fig:feasible_region2_lrl}}
      \hspace{2cm}
      \subfigure[Completing to a loop]{\includegraphics{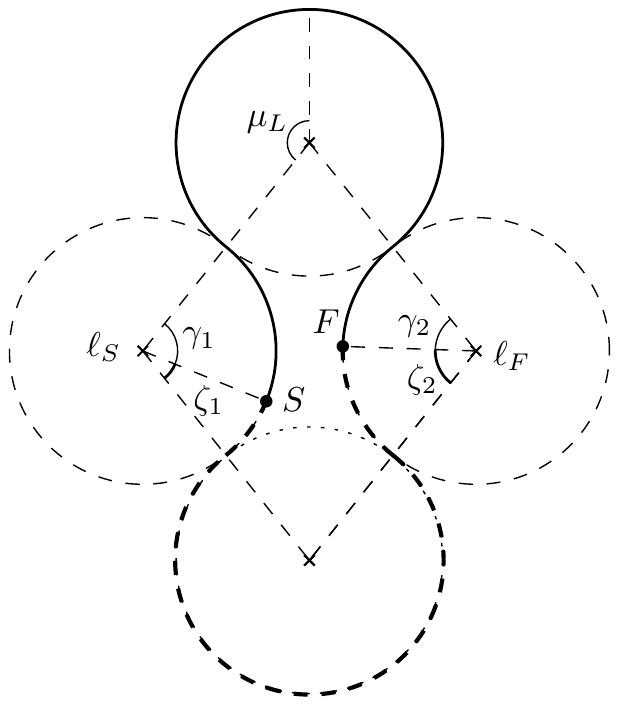}%
	\label{fig:region2_loop}}}
    \caption{Proof of Lemma~\ref{lem:lrl-length-specific} for case~A}
  \end{figure}

  For \prlr-paths, we could argue analogously, or we can simply observe
  that 
  \[
  \lrlr(\sigma,\delta) = \llrl(\pi-\sigma,\pi-\delta) =
  4\mul(\pi-\sigma,\pi-\delta) +2(\pi-\delta) - 2\pi =
  4\mur(\sigma,\delta) - 2\delta.
  \]

  Assume now that $(\sigma, \delta) \in B$. By
  Lemma~\ref{lem:case-b-basics}, we have $\pi < 2\delta \leq 2\pi$.
  By Lemma~\ref{lem:ccc-endpoint-locations} and Eq.~\eqref{eq:gamma1} we
  have
  \begin{align*}
    2\mul-\pi & \leq \gamma_1 < 2\pi \\
    0 & \leq \gamma_2 \leq 2\mul-\pi
  \end{align*}
  Together these imply $-\pi \leq \gamma_1+\gamma_2-2\mul < \pi$.
  Since $\gamma_1+\gamma_2-2\mul \equiv 2\delta \pmod{2\pi}$, we must
  have $\gamma+\gamma_2-2\mul = 2\delta-2\pi$.  It follows that
  $\llrl(\sigma, \delta) = 4\mul+2\delta-2\pi$.

  We turn to the \prlr-path in case~B.  By
  Lemma~\ref{lem:ccc-endpoint-locations} and Eq.~\eqref{eq:gamma2} we have
  (here, $\gamma_1$ and $\gamma_2$ are the right-turning arcs)
  \begin{align*}
    0 & \leq \gamma_1 \leq 2\mur-\pi \\
    2\mur-\pi & \leq \gamma_2 < 2\pi
  \end{align*}
  Together we have $-\pi \leq \gamma_1 + \gamma_2 -2\mur \leq \pi$.
  Since $-\gamma_1-\gamma_2+2\mur \equiv 2\delta \pmod{2\pi}$, we must
  have $-\gamma_1 + 2\mur - \gamma_2 = 2\delta - 2\pi$, which shows
  that $\lrlr = \gamma_1+2\mur+\gamma_2 = 4\mur-(2\delta-2\pi) =
  4\mur-2\delta+2\pi$.
\end{proof}

\paragraph{Proof of Lemma~\ref{lem:lrl-changes-alpha-beta}}

\begin{proof}
  We first prove that for $(\alpha, \beta) \in \BD$, the function~$\alpha
  \mapsto \llrl(\alpha, \beta)$ is decreasing.
  By Lemma~\ref{lem:lrl-length-specific} we have
  $\llrl(\alpha,\beta)=4\mul+2\delta-2\pi$ in~$\BD$.  We have $2\delta =
  \beta-\alpha$, $\mul=\pi-\arcsin({\dl}/{4})$, $\dl^2 =
  (d-\sin\beta+\sin\alpha)^2 + (\cos\beta-\cos\alpha)^2$
  by Eq.~\eqref{eq:dl-alpha-beta}, and
  $\frac{\partial\arcsin{x}}{\partial{x}} = \frac{1}{\sqrt{1-x^2}}$.
  Setting $D_\tL = \dl\sqrt{1-({\dl/4})^{2}}$,
  we have
  \begin{align*}
    \frac{\partial\llrl}{\partial\alpha}(\alpha,\beta) & = 
    -\frac{(d-\sin\beta+\sin\alpha)\cos\alpha +
      (\cos\beta-\cos\alpha)\sin\alpha}{D_L} - 1.
  \end{align*}
  We also have
  \begin{align*}
    (d-\sin\beta+\sin\alpha)\cos\alpha + (\cos\beta-\cos\alpha)\sin\alpha
    = & d\cos\alpha - \sin\beta\cos\alpha + \cos\beta\sin\alpha \\
    = & d\cos\alpha - \sin(\beta-\alpha) \geq 0.
  \end{align*}
  The last inequality holds since $0 \leq \alpha \leq \pi/2$ and
  $\beta-\alpha \geq \pi$ in~$\BD$ by Lemma~\ref{lem:case-b-basics}.
  It follows that $\frac{\partial\llrl}{\partial\alpha}(\alpha, \beta)
  \leq -1 <0$.

  We now prove that if $(\alpha, \beta) \in \BD$ and $\beta \geq 3\pi/2$ then 
  $\frac{\partial\llrl}{\partial\beta}(\alpha, \beta) \geq 1$.
  We have
  \begin{align*}
    \frac{\partial\llrl}{\partial\beta}(\alpha,\beta) & = 
    \frac{(d-\sin\beta+\sin\alpha)\cos\beta +
      (\cos\beta-\cos\alpha)\sin\beta}{D_L} + 1.
  \end{align*}
  %
  $\beta \geq {3\pi}/{2}$
  implies that $\sin\beta \leq 0$ and $\cos\beta \geq
  0$. By Lemma~\ref{lem:case-b-basics} we have $\pi+\alpha \leq \beta \leq
  2\pi-\alpha$, which implies that $\cos\beta \leq \cos\alpha$.  It
  follows that the first term of
  $\frac{\partial\llrl}{\partial\beta}(\alpha,\beta)$ is nonnegative,
  proving that $\frac{\partial\llrl}{\partial\beta}(\alpha, \beta)
  \geq 1$.

  We finally prove that for $(\alpha, \beta) \in \AD \cup \BD$ the
  function $\beta \mapsto \llrl(\alpha, \beta)$ is increasing.  Let $U
  = d-\sin\beta+\sin\alpha$, and $V = \cos\beta-\cos\alpha$ (note that
  $\dl^2 = U^2+V^2$).  If $U\cos\beta+V\sin\beta \geq 0$, we have
  $\frac{\partial\llrl}{\partial\beta}(\alpha,\beta) > 0$, and $\beta
  \mapsto \llrl(\alpha,\beta)$ is increasing. Now let us assume that
  $U\cos\beta+V\sin\beta$ is negative, and let us compare the squared
  terms of $\frac{\partial\llrl}{\partial\beta}$ (we only consider the
  numerator since the denominator is always positive).
  \begin{align*}
    (U\cos\beta+V\sin\beta)^2 - \dl^2\pth{1-\frac{\dl^2}{16}}
    = & (U^2\cos^2\beta+V^2\sin^2\beta+2UV\cos\beta\sin\beta) - (U^2+V^2) +
    \frac{(U^2+V^2)^2}{16} \\
    = & (-U^2\sin^2\beta-V^2\cos^2\beta+2UV\cos\beta\sin\beta) +
    \frac{(U^2+V^2)^2}{16} \\
    = & \frac{1}{16}\pth{-16(U\sin\beta-V\cos\beta)^2+(U^2+V^2)^2}.
  \end{align*}

  For $(\alpha, \beta) \in \AD \cup \BD$, we have $\dlr <
  2$. From Eq.~\eqref{eq:dlr-alpha-beta}, $\dlr^2 =
  (d+\sin\beta+\sin\alpha)^2 + (\cos\beta+\cos\alpha)^2 < 4$. By
  substituting $d+\sin\alpha = U+\sin\beta$ and $\cos\alpha =
  \cos\beta-V$, we have
  \begin{align*}
    & \dlr^2 < 4 \\
    & \Leftrightarrow (U+2\sin\beta)^2 + (2\cos\beta-V)^2 < 4 \\
    & \Leftrightarrow
    U^2+V^2+4(U\sin\beta-V\cos\beta)+4(\sin^2\beta+\cos^2\beta) < 4 \\
    & \Leftrightarrow (U^2+V^2)+4(U\sin\beta-V\cos\beta) < 0.
  \end{align*}

  Since $U^2+V^2 = \dl^2 \geq 0$, this implies that
  $U\sin\beta-V\cos\beta < 0$, and further the squared term
  $16(U\sin\beta-V\cos\beta)^2$ is greater than $(U^2+V^2)^2$, which
  implies that $(U^2+V^2)^2-16(U\sin\beta-V\cos\beta)^2 < 0$,
  completing the proof.
\end{proof}


\paragraph{Proof of Lemma~\ref{lem:l-r-fixed}}

\begin{proof}
  Setting $D_\tL = \dl\sqrt{1-({\dl/4})^{2}}$ and
  $D_\tR = \dr\sqrt{1-({\dr/4})^{2}}$, 
  we obtain the derivatives of $\lL$ and $\lR$ using $\mul = \pi -
  \arcsin ({\dl}/{4})$, $\mur = \pi - \arcsin({\dr}/{4})$, and
  Eqns.~\eqref{eq:dl-sigma-delta}--\eqref{eq:drl-sigma-delta}:
  \begin{align}
    \frac{\partial\lL}{\partial\delta}(\sigma,\delta) & = 
    2 + \frac{-4\cos\delta\sin\delta +
      2d\cos\delta\cos\sigma}{D_\tL},\label{eq:l-derivative-delta} \\
    \frac{\partial\lL}{\partial\sigma}(\sigma,\delta) & = 
    \frac{-2d\sin\delta\sin\sigma}{D_\tL},\label{eq:l-derivative-sigma}\\
    \frac{\partial\lR}{\partial\delta}(\sigma,\delta) & = 
    -2 + \frac{-4\cos\delta\sin\delta -
      2d\cos\delta\cos\sigma}{D_\tR},\label{eq:r-derivative-delta}\\
    \frac{\partial\lR}{\partial\sigma}(\sigma,\delta) & = 
    \frac{2d\sin\delta\sin\sigma}{D_\tR}.\label{eq:r-derivative-sigma}
  \end{align}
  The derivatives are not defined when $D_\tL = 0$ or $D_\tR = 0$.
  $D_\tL=0$ occurs when $\dl = 0$ or $\dl = 4$, and $D_\tR = 0$ occurs
  when $\dr = 0$ or $\dr = 4$.

  In the interior of~$\RECTII$, we have $\sin\delta > d/2$, which
  implies $\cos\delta^2 < 1-d^2/4$. So for $\as < \delta \leq \pi/2$,
  we have $\drl < 2$ by Eq.~\eqref{eq:drl-sigma-delta}, and for $\pi/2
  \leq \delta < \pi-\as$, we have $\dlr < 2$
  by Eq.~\eqref{eq:dlr-sigma-delta}.  By the triangle inequality it
  follows that $\dl < 4$ and $\dr < 4$. Also,
  by Eq.~\eqref{eq:dl-sigma-delta}
  $\dl^2=(d-2\sin\delta\cos\sigma)^2+4\sin\delta(-\cos\sigma+1)$. Since
  $\sin\delta > 0$, $\dl = 0$ can occur only when $\sigma = 0$ and
  $\delta \in \{\as, \pi-\as \}$, which occurs at a corner
  of~$\RECTII$. Similar arguments hold for the case where
  $\dr^2=0$. Thus, in the interior of~$\RECTII$, $D_\tL \neq 0$ and
  $D_\tR \neq 0$.

  Consider Eqns.~\eqref{eq:l-derivative-sigma}
  and~\eqref{eq:r-derivative-sigma}. Since $\sin\delta >0$ in
  $\RECTII$ and $\sin\sigma \geq 0$ with equality only for $\sigma \in
  \{0, \pi\}$, we have $\frac{\partial\lL}{\partial\sigma} < 0$ and
  $\frac{\partial\lR}{\partial\sigma} > 0$ in the interior
  of~$\RECTII$.

  It remains to discuss the two functions of~$\delta$.  We show that
  $\frac{\partial\lL}{\partial\delta}(\sigma, \delta) > 0$. If $Z =
  2\cos\delta(-2\sin\delta + d\cos\sigma) \geq 0$, this is true. Let
  us thus assume that $Z < 0$. Since $\sin\delta \geq d/2$ implies
  that $-2\sin\delta+d\cos\sigma \leq -2(d/2)+d\cos\sigma =
  d(-1+\cos\sigma) \leq 0$, we have $\cos\delta >0$ and so
  $-2\sin\delta+d\cos\sigma <0$.  $\sin\delta \geq d/2$ also implies
  that $d^2 \leq 4\sin^2\delta$, so by Eq.~\eqref{eq:dl-sigma-delta}
  $\dl^2 = d^2 - 4d\sin\delta\cos\sigma + 4\sin^2\delta \leq
  8\sin^2\delta - 4d\sin\delta\cos\sigma =
  4\sin\delta(2\sin\delta-d\cos\sigma)$. Since $\dl^2 \geq 0$, we have
  \begin{align}
    \label{eq:dl-smaller}
    \dl^4 \leq 16\sin^2\delta(2\sin\delta-d\cos\sigma)^2.
  \end{align}
  On the other hand, we have 
  \begin{align}
    \label{eq:dl-larger}
    \dl^2  &= d^2 - 4d\sin\delta\cos\sigma + 4\sin^2\delta \nonumber\\
    & \geq 
    d^2\cos^2\sigma - 4d\sin\delta\cos\sigma + 4\sin^2\delta =
    (2\sin\delta-d\cos\sigma)^2. 
  \end{align}
  Now we want to show that $2D_\tL \geq -Z$, which will imply
  $\frac{\partial\lL}{\partial\delta}(\sigma, \delta) \geq 0$. Let us
  compare the squared terms:
  \begin{align*}
    & 4{D_\tL}^2 - Z^2   \\
    &= 4D_\tL^2 - 4\cos^2\delta(-2\sin\delta+d\cos\sigma)^2 \\ 
    &= 4\pth{
      \dl^2\pth{1-\frac{\dl^2}{16}}
      -\cos^2\delta(2\sin\delta-d\cos\sigma)^2}\\
    & \geq 4\pth{
      \dl^2 - \sin^2\delta(2\sin\delta-d\cos\sigma)^2
      - \cos^2\delta(2\sin\delta-d\cos\sigma)^2} \quad
    \text{(by~\eqref{eq:dl-smaller})} \\
    & = 4\pth{
      \dl^2 - (2\sin\delta-d\cos\sigma)^2}  \geq 0 \quad
    \text{(by~\eqref{eq:dl-larger})}.
  \end{align*}
  One of the inequalities in above formula is a strict inequality:
  if~\eqref{eq:dl-larger} is an equality, then $\dl^2 = (2\sin\delta -
  d\cos\sigma)^2$, which means that $\cos^2\sigma = 1$. This implies
  that $d^2 < 4\sin^2\delta$ since we have $-2\sin\delta + d\cos\sigma
  < 0$, so~\eqref{eq:dl-smaller} is strict. In the case
  where~\eqref{eq:dl-smaller} is an equality, we can argue similarly
  that~\eqref{eq:dl-larger} is strict.

  Similarly, we prove that $\frac{\partial\lR}{\partial\delta}(\sigma,
  \delta) < 0$, since $\sin\delta \geq d/2$ again implies that $\dr^2
  = d^2 + 4d\sin\delta\cos\sigma + 4\sin^2\delta \leq
  4\sin\delta(2\sin\delta + d\cos\sigma)$, so
  \begin{align*}
    \dr^4 & \leq 16\sin^2\delta(2\sin\delta+d\cos\sigma)^2,
  \end{align*}
  and we have 
  \begin{align*}
    \dr^2 & \geq d^2\cos^2\sigma + 4d\sin\delta\cos\sigma + 4\sin^2\delta =
    (2\sin\delta+d\cos\sigma)^2.\qedhere
  \end{align*}
\end{proof}

\paragraph{Proof of Lemma~\ref{lem:lrl-eq-rlr}}

\begin{proof}
   By Lemma~\ref{lem:l-r-fixed}, neither $\lL$ nor $\lR$ has a local
   extremum in the interior of~$\RECTII$, so any local extremum of
   $\lC(\sigma,\delta)$ must be a point in the set~$\LER$ of points
   $(\sigma,\delta)$ with $\lL(\sigma,\delta) = \lR(\sigma,\delta)$. By
   Lemma~\ref{lem:l-r-fixed}, $\LER$ is a $\delta$-monotone curve.
   Since $\lL(\sigma,\delta) = \lR(\pi-\sigma,\pi-\delta)$, the
   curve~$\LER$ passes through the point~$(\pi/2, \pi/2)$. By
   Lemma~\ref{lem:l-r-fixed}, this implies that $\lL(\sigma, \delta) <
   \lR(\sigma, \delta)$ for the quadrant~$\pi/2 \leq \sigma \leq \pi$,
   $\as \leq \delta \leq \pi/2$, and that $\lR(\sigma, \delta) <
   \lL(\sigma, \delta)$ for the quadrant~$0 \leq \sigma \leq \pi/2$,
   $\pi/2 \leq \delta \leq \pi-\as$ except at the point~$(\pi/2,
   \pi/2)$.  By point symmetry, we can restrict our attention to the
   range $\pi/2<\sigma< \pi$, $\pi/2< \delta < \pi-\as$.
 
   Assume for a contradiction that $(\sigma,\delta) \in \LER$ is a
   local extremum of~$\lL$, restricted to~$\LER$.  This implies that
   the gradient $\nabla \lL(\sigma,\delta)$ and the normal of $\LER$ in
   $(\sigma,\delta)$ are linearly dependent, by the method of Lagrange
   Multipliers.  The normal of $\LER$ is the gradient of
   $\lL(\sigma,\delta) - \lR(\sigma,\delta)$, so
   $\nabla\lL(\sigma,\delta)$ and $\nabla\lR(\sigma,\delta)$ must be
   linearly dependent. 
 
   For the two vectors to be linearly dependent, we would have to have
   \[
   D_\tL\frac{\partial\lL}{\partial\delta}(\sigma,\delta) + 
   D_\tR\frac{\partial\lR}{\partial\delta}(\sigma,\delta) = 0,
   \]
   which means
   \[
   2D_\tL -2D_\tR - 8 \cos\delta\sin\delta = 0.
   \]
   In the range under consideration, $-8\cos\delta\sin\delta > 0$. We
   will show that $D_\tL > D_\tR$, a contradiction.  We have
   \begin{align*}
     16(D_\tL^{2} - D_\tR^{2}) & = \dl^{2}(16-\dl^{2}) - 
     \dr^{2}(16-\dr^{2}) \\
 	& = \dr^4 - \dl^4 + 16(\dl^2-\dr^2) \\
 	& = (\dl^2-\dr^2)(16-(\dl^2+\dr^2)) \\
     & = -8d\cos\sigma\sin\delta(16 - (2d^{2} + 8\sin^{2}\delta)).
   \end{align*}
   Since $\cos\sigma < 0$ and $d < 2$, the expression is positive.
 \end{proof}

\paragraph{Proof of Lemma~\ref{lem:a-monotonicity}}

\begin{proof}
  We will show below that the two functions $d \mapsto \lA(d)$ and $d \mapsto
  \lA(d)-d$ have no extremum on the interval~$(0, \sqrt{2})$.  This will imply
  the claim if we observe
  that
  \begin{align*}
    \lA(0) & = 7\pi/3, \\
    \lA(\sqrt{2}) & = 5\pi/2 > \lA(0) \\
    \lA(\sqrt{2}) - \sqrt{2} & = 5\pi/2 - \sqrt{2} < \lA(0) - 0.
  \end{align*}    

  Again we will employ Lagrange multipliers.  Let us first give the necessary
  derivatives.  Setting $\sigma = \pi$ and $\pi/2 < \delta \leq \pi-\as$ (by
		  Lemma~\ref{lem:max-rectii}), we have:
  \begin{align*}
    \dl & = d + 2 \sin \delta \qquad \text{by
      (\ref{eq:dl-sigma-delta})}\\
    \dr & = 2 \sin \delta - d \qquad \text{by
      (\ref{eq:dr-sigma-delta})}\\
    \lL(d, \pi, \delta) &= 2\pi - 4\arcsin(\dl/4) + 2\delta \qquad \text{by
      (\ref{eq:def-l})}\\
    \lR(d, \pi, \delta) &= 4\pi - 4\arcsin(\dr/4) - 2\delta \qquad \text{by
      (\ref{eq:def-r})}
  \end{align*}
  Let us introduce $F_\tL = \sqrt{1-(\dl/4)^{2}}$ and $F_\tR =
  \sqrt{1-(\dr/4)^{2}}$ to obtain:
  \begin{align*}
    \frac{\partial}{\partial d}\lL(d, \pi, \delta)
    & = - \frac{1}{F_\tL} \qquad
    \frac{\partial}{\partial \delta}\lL(d, \pi, \delta)
     = 2 - \frac{2\cos\delta}{F_\tL} \\
    \frac{\partial}{\partial d}\lR(d, \pi, \delta)
    & = \frac{1}{F_\tR} \qquad
    \frac{\partial}{\partial \delta}\lR(d, \pi, \delta)
     = -2 - \frac{2\cos\delta}{F_\tR}
  \end{align*}

  We consider first the function $d \mapsto \lA(d)$.  An extremum of
  $\lA(d)$ is an extremum of the two-parameter function $(d, \delta)
  \mapsto \lL(d, \pi, \delta)$ under the restriction $\lL(d, \pi,
  \delta) = \lR(d, \pi, \delta)$.  Such an extremum would have to
  satisfy the condition $\nabla \lL(d, \pi, \delta) = \lambda
  \nabla\lR(d, \pi, \delta)$. For this to hold:
  \begin{align}
    \lambda = -\frac{F_\tR}{F_\tL} 
	& = \frac{2-2\cos\delta/F_\tL}{-2-2\cos\delta/F_\tR}, \nonumber
  \end{align}
  which implies $F_\tR-F_\tL = -2\cos\delta$.  Since $\dr < 2$, $F_\tR
  > \sqrt{3}/2$, and since $\dl < 2+\sqrt{2} < 2\sqrt{3}$, $F_\tL >
  1/2$.  This implies $F_\tR + F_\tL > 1$, and so
  \begin{align}
    -2\cos\delta = F_\tR - F_\tL < (F_\tR - F_\tL)(F_\tR + F_\tL) =
    F_\tR^2-F_\tL^2 = \frac{d}{2}\sin\delta. \label{eq:FR-FL}
  \end{align}
  The condition $\lL(d, \pi, \delta) = \lR(d, \pi, \delta)$ implies
  (using $\arcsin x - \arcsin y \geq x - y$ for $0 \leq y \leq x \leq
  1$)
  \[
  4\delta = 2\pi + 4\big(\arcsin\frac{\dl}{4} - \arcsin\frac{\dr}{4}\big)
  \geq 2\pi + (\dl - \dr) = 2\pi + 2d.
  \]
  It follows that $\pi/2 + d/2 \leq \delta \leq \pi$, and so we have
  $\sin\delta \leq \cos\frac{d}{2}$ and $-\cos\delta \geq \sin\frac{d}{2}$.
  With~(\ref{eq:FR-FL}) this gives
  \[
  2\sin\frac{d}{2} \leq -2\cos\delta < \frac{d}{2} \sin\delta 
  \leq \frac{d}{2}\cos \frac{d}2.
  \]
  Setting $f(x) := x\cos x - 2\sin x$, this implies $f(d/2) > 0$.  But
  this is impossible, since $f(0) = 0$ and $f'(x) = -2\cos x - x\sin x
  \leq 0$ for $0 \leq x \leq d/2 < \pi/2$.

  \medskip

  Consider next the function $d \mapsto \lA(d) - d$.  An extremum of
  $\lA(d) - d$ is an extremum of the two-parameter function $(d, \delta)
  \mapsto \lL(d, \pi, \delta) - d$ under the restriction $\lL(d, \pi,
  \delta) = \lR(d, \pi, \delta)$.  Such an extremum would have to
  satisfy the condition $\lambda \nabla (\lL(d, \pi, \delta) -  d) =
  \nabla\big(\lL(d, \pi, \delta) - \lR(d, \pi, \delta)\big)$, or
  \begin{align*}
    \lambda ( - \frac{1}{F_\tL} - 1) & = 
    - \frac{1}{F_\tL} - \frac{1}{F_\tR}, \\
    \lambda (2 - \frac{2 \cos \delta}{F_\tL})  & = 
    4 - \frac{2 \cos \delta}{F_\tL} + \frac{2 \cos \delta}{F_\tR}
  \end{align*}
  The two components give us the following conditions on~$\lambda$:
  \begin{align*}
    \lambda & = 1 + \frac{1/F_\tR - 1}{1/F_\tL + 1} 
    = 1 + \frac{2 + (2\cos\delta)/F_\tR}{2 - (2\cos\delta)/F_\tL}.
  \end{align*}
  This is equivalent to 
  \begin{align*}
    \Big( 2 - \frac{2\cos\delta}{F_\tL}\Big)
    \Big( \frac{1}{F_\tR} - 1 \Big)
    & = 
    \Big( 2 + \frac{2\cos\delta}{F_\tR}\Big)
    \Big( \frac{1}{F_\tL} + 1 \Big).
    \intertext{Multiplying out and rearranging the terms gives}
	(2-2\cos\delta)\Big(\frac{1}{F_\tR} - \frac{1}{F_\tL}\Big)
	& = 4\Big( \frac{\cos\delta}{F_\tR F_\tL} + 1 \Big).
  \end{align*}
  Since $\dl - \dr = 2d > 0$, we have $F_\tR > F_\tL$.  With
  $\cos\delta < 0$ this implies that the left-hand side is negative.
  We will now show that the right-hand side is non-negative, a
  contradiction, and so $d \mapsto \lA(d) - d$ cannot have a local
  extremum.

  It is enough to show that $\cos^2\delta \leq (F_\tR F_\tL)^2$:
  \begin{align*}
    (F_\tR F_\tL)^2 - \cos^2\delta
    & = \Big(1-\frac{\dl^2 + \dr^2}{16} + \frac{(\dl\dr)^2}{16^2}
    \Big)  - \cos^2\delta\\ 
    & = -\frac{2d^2+8\sin^2\delta}{16} + \frac{(\dl\dr)^2}{16^2} +
    \sin^2\delta \\ 
    & = \frac{-2d^2+8\sin^2\delta}{16} + \frac{(\dl\dr)^2}{16^2} \geq 0.
  \end{align*}
  The last inequality holds since $\sin\delta \geq d/2$.
\end{proof}


\paragraph{Proof of Lemma~\ref{lem:rsl-changes-alpha-beta}}

\begin{proof}
  The two derivatives in Eqns.~\eqref{eq:rsl-dalpha}
  and~\eqref{eq:rsl-dbeta} are defined and both positive in the
  interior of~$\Bc$, implying the first two claims.

  For the third claim, we need to show that 
  \begin{align*}
    0 & \leq \frac{\partial}{\partial\alpha}\lrsl(\alpha, \beta) -  
    \frac{\partial}{\partial\beta}\lrsl(\alpha, \beta) 
    = -\cos\gamma_{\tR} + \cos\gamma_{\tL}
  \end{align*}
  for $\beta = 2\pi - \alpha$. By
  Lemma~\ref{lem:ccc-endpoint-locations}, $S$ lies on the
  arc~$\arc{\SL_0\SL_1}$ of~$L_S$.  If $\beta=2\pi-\alpha$, then
  $F$~lies on the counter-clockwise arc~$\arc{\FL_1\FL_0}$ of~$L_F$.
  The two circular arcs of the \prsl-path have two components, namely,
  $\gamma_{\tR} = \alpha+\zeta$ and $\gamma_{\tL} = \beta+\zeta = 2\pi
  -\alpha+\zeta$, see Figure~\ref{fig:along-diagonal}. This implies
  \begin{figure}
    \centerline{\includegraphics{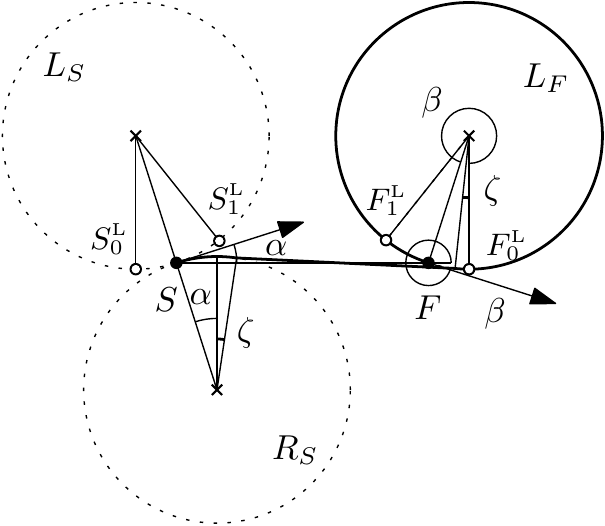}}
    \caption{An \prsl-path for $\beta=2\pi-\alpha$}
    \label{fig:along-diagonal}
  \end{figure}
  \begin{align*}
    -\cos\gamma_{\tR} + \cos\gamma_{\tL} 
    & = - \cos(\alpha + \zeta) + \cos(\alpha - \zeta)
    = 2\sin\alpha\sin{\zeta} \geq 0. 
  \end{align*}

  The second and third claim immediately imply the last one.
\end{proof}

\paragraph{Proof of Lemma~\ref{lem:region3-lrl-rsl-shorter}}

\begin{proof}
  We first compare the \plrl-path and the \plsl-path. Let $\gamma_1$
  and $\gamma_2$ be the two arcs on the \plrl-path.  By
  Lemma~\ref{lem:ccc-endpoint-locations}, we have $2\mul-\pi \leq
  \gamma_1 \leq \mul$ and $0 \leq \gamma_2 \leq 2\mul-\pi$.  The
  \plsl-path has length $\llsl = \gamma_1 + \gamma_2 + 2(2\pi-\mul) +
  \dl$.  We thus have $\llsl -\llrl = \dl + 4(\pi - \mul) \geq 0$.

  Consider now the \plrl-path and the
  \prlr-path. By~\eqref{eq:dl-sigma-delta}
  and~\eqref{eq:dr-sigma-delta}, we have $\dr^{2} - \dl^{2} =
  8d\sin\delta\cos\sigma \leq 0$ for $(\sigma, \delta)\in \BD$, and so
  $\dr \leq \dl$, implying~$\mul \leq \mur$.  By
  Lemma~\ref{lem:lrl-length-specific}, we have $\lrlr - \llrl = 4(\mur - \mul -
  \delta + \pi) \geq 0$.
  
  Finally, we compare \prsl-path and \prsr-path.  The \prsl-path
  consist of an initial right-turning arc $\arc{ST_S}$, a segment
  $\overline{T_S T_F}$, and a final left-turning arc~$\arc{T_F F}$.
  The \prsr-path consist of an initial right-turning arc $\arc{SR_1}$,
  a segment $\overline{R_1 R_2}$, and a final right-turning arc
  $\arc{R_2 F}$, see Figure~\ref{fig:dlr2_rsr_ge_rsl}.
  \begin{figure}
    \centerline{\includegraphics{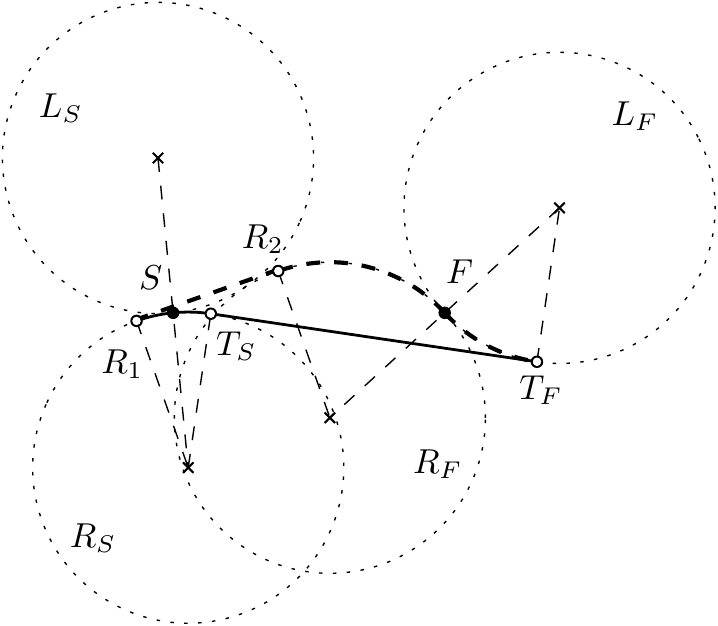}}
    \caption{$\lrsl(\alpha, \beta) \leq
      \lrsr(\alpha, \beta)$ for $(\alpha, \beta) \in \BD$ with $\beta$
      close to $\blr(\alpha)$.}
    \label{fig:dlr2_rsr_ge_rsl}
  \end{figure}
  
  We first claim that the arc~$\arc{ST_S}$ is common to both paths.
  Indeed, since $\dlr < 2$, by Lemma~\ref{lem:ccc-endpoint-locations},
  the initial arc~$\arc{SR_1}$ of the \prsr-path must have length at
  least~$\pi$ (Figure~\ref{fig:rlr-regions}), while the
  arc~$\arc{ST_S}$ must be shorter than~$\pi$
  (Figure~\ref{fig:lrl-regions}).
  
  Since $|\arc{T_FF}| = 2\pi-|\arc{FT_F}|$ and $|\arc{T_S R_1}| =
  2\pi-|\arc{R_1 T_S}|$, where $\arc{FT_F}$ is a left-turning arc on
  the disk $L_F$ and $\arc{R_1 T_S}$ is a right-turning arc on disk
  $R_S$, we have
  \begin{align*}
    \lrsl - \lrsr & = 
    \big(|\overline{T_S T_F}| - |\arc{FT_F}|\big) - 
    \big(-|\arc{R_1 T_S}| + |\overline{R_1 R_2}| + |\arc{R_2 F}|\big) \\
    & = |\arc{R_1 T_S} \cup \overline{T_S T_F}| - 
    |\overline{R_1 R_2} \cup \arc{R_2 F} \cup \arc{FT_F}|.
  \end{align*}
  The path $\overline{R_1 R_2} \cup \arc{R_2 f} \cup \arc{FT_F}$ is a
  path connecting~$R_1$ with~$T_F$ while avoiding the interior
  of~$R_S$. However, the path $\arc{R_1 T_S} \cup \overline{T_S T_F}$
  is clearly the shortest path of this kind, and so $\lrsl - \lrsr
  \leq 0$.
\end{proof}

\paragraph{Proof of Lemma~\ref{lem:lrl-rsl-monotone}}

\begin{proof}
  We first observe that the function $\alpha \mapsto
  \llrl(\alpha,\beta) - \lrsl(\alpha, \beta)$ is decreasing. This
  follows immediately from Lemmas~\ref{lem:rsl-changes-alpha-beta}
  and~\ref{lem:lrl-changes-alpha-beta}.  We also claim that the function
  $\beta \mapsto \llrl(\alpha,\beta) - \lrsl(\alpha, \beta)$ is
  increasing for $\beta\geq 3\pi/2$.  Together, these facts prove the
  lemma: consider two values $0 \leq \alpha_1 < \alpha_2 \leq \as$.
  Since $\blr(\alpha)$ is a decreasing function, we have
  $\blr(\alpha_1) > \blr(\alpha_2) \geq 2\pi - \as > 3\pi/2$, and so
  \begin{align*}
    \llrl(\alpha_1, \blr(\alpha_1)) - \lrsl(\alpha_1, \blr(\alpha_1)) &
    \geq 
    \llrl(\alpha_1, \blr(\alpha_2)) - \lrsl(\alpha_1, \blr(\alpha_2)) \\
    & \geq \llrl(\alpha_2, \blr(\alpha_2)) - \lrsl(\alpha_2, \blr(\alpha_2)).
  \end{align*}

  It remains to show that for $(\alpha, \beta) \in \BD$ with $\beta
  \geq {3\pi}/{2}$, the function $\beta \mapsto
  \llrl(\alpha,\beta) - \lrsl(\alpha, \beta)$ is increasing.  Since we
  are in case~$\BD$, by Lemma~\ref{lem:ccc-endpoint-locations} the
  point $S$ lies on the arc~$\arc{\SL_0\SL_1}$ of~$L_S$, while $F$
  lies on the arc~$\arc{\FL_2\FL_1}$ of~$L_F$ (see
  Figure~\ref{fig:lrl-regions}). It follows that $\gamma_\tL > \beta
  \geq 3\pi/2$, and so $\cos\gamma_\tL > 0$.  This implies that
  $\frac{\partial}{\partial\beta}\lrsl(\alpha, \beta) = 1
  -\cos\gamma_\tL < 1$.  On the other hand, by
  Lemma~\ref{lem:lrl-changes-alpha-beta}, $\frac{\partial}{\partial\beta}
  \llrl(\alpha,\beta) \geq 1$ for $\beta \geq {3\pi}/{2}$, and the
  claim follows.
\end{proof}

\paragraph{Proof of Lemma~\ref{lem:blr-no-extremum}}

\begin{proof}
  We first claim that for $(\alpha, \beta)\in \Bc$ and $d < 2$, if
  $\gamma_{\tR} + \gamma_{\tL} \leq 2\pi$, then $\lrsl(d, \alpha,
  \beta) \leq d+2\pi$.

  Let $T_S$ and $T_F$ denote the points of tangency of the $S$-segment
  to~$R_S$ and~$L_F$.  We observed above that $\gamma_{\tR} \leq \pi$.
  If we also have $\gamma_{\tL} \leq \pi$ then $d \geq |T_S T_F|$, and
  the claim follows immediately.
  
  If $\gamma_{\tR} \geq \pi/2$, then we have $\gamma_{\tL} \leq 2\pi -
  \gamma_{\tR} \leq 3\pi/2$.  But then $d = |SF| \geq 2$, a
  contradiction.  It follows that we must have $\gamma_{\tR} < \pi/2$.
  See Figure~\ref{fig:arc-sum-less-2pi}.  
  \begin{figure}
    \centerline{\includegraphics{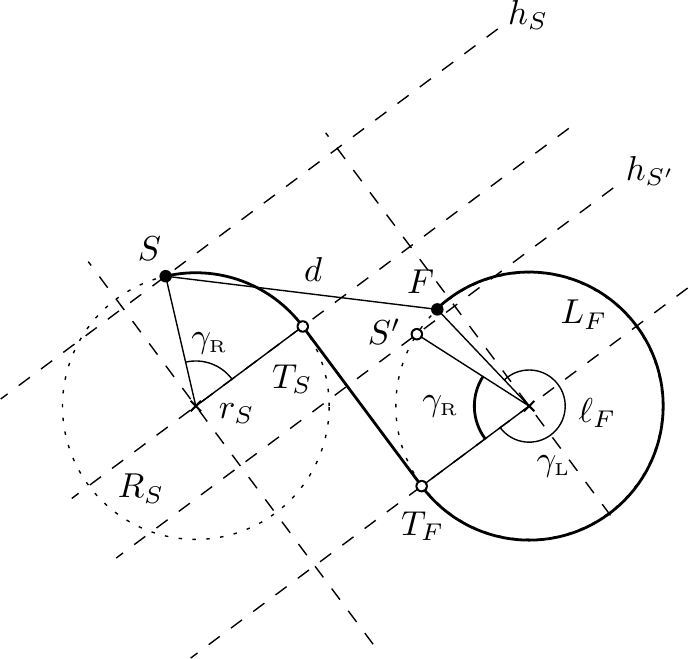}}
    \caption{If $\gamma_{\tR} + \gamma_{\tL} \leq 2\pi$ then $\lrsl \leq
      d+2\pi$}
    \label{fig:arc-sum-less-2pi}
  \end{figure}
  Let $S'$ be the point on~$L_F$ such that the counter-clockwise
  arc~$\arc{S'T_F}$ on~$L_F$ has length~$\gamma_{\tR}$.  Let $h_S$ and
  $h_{S'}$ be lines through $S$ and~$S'$ orthogonal to the
  segment~$\overline{T_S T_F}$.  The distance between $h_S$
  and~$h_{S'}$ is $|T_S T_F|$, and so we have $|S S'| \geq |T_S T_F|$.
  By the triangle inequality, we have $d + |\arc{FS'}| = |SF| +
  |\arc{FS'}| \geq |SS'| \geq |T_S T_F|$, so $|T_S T_F| - |\arc{FS'}|
  \leq d$. Since $2\pi - |\arc{FS'}| = \gamma_{\tR} + \gamma_{\tL}$,
  we have
  \[
  \lrsl(d, \alpha,\beta) = \gamma_\tR + |T_S T_F| + \gamma_{\tL} 
  = |T_S T_F| + 2\pi - |\arc{FS'}| \leq d + 2\pi, 
  \]
  and the claim follows.

  Now let $\alpha_0$ be such an extremum, and consider the \prsl-path for
  $(\alpha_0, \blr(\alpha_0))$.  If $\gamma_{\tR} + \gamma_{\tL} \leq
  2\pi$, then by above argument we have $\lrsl(d,
  \alpha_0, \blr(\alpha_0)) \leq d + 2\pi$.  We can therefore assume
  $\gamma_{\tR} + \gamma_{\tL} > 2\pi$.  The point $(\alpha_0,
  \blr(\alpha_0))$ is an extremum of the function $\lrslp(\alpha,
  \beta)$, under the constraint that $\dlr^{2} = 4$, and so there must
  be a constant~$\lambda$ such that $\nabla \lrslp(\alpha_0,
  \blr(\alpha_0)) = \lambda \nabla \dlr^{2}(\alpha_0,
  \blr(\alpha_0))$.

  Using~$\pi/2 < \sigma, \delta \leq \pi$
  and~\eqref{eq:dlr-sigma-delta} we have
  \begin{align}
    \label{eq:dlr-dsigma}
    \frac{\partial}{\partial \sigma} \dlr^{2} 
    & = 4d\cos\delta \cos\sigma > 0 \\
    \label{eq:dlr-ddelta}
    \frac{\partial}{\partial \delta} \dlr^{2}
    & = -4\sin\delta(d\sin\sigma+2\cos\delta). \\
    \intertext{Using~\eqref{eq:rsl-dalpha} and~\eqref{eq:rsl-dbeta} we get}
    \label{eq:rsl-dsigma}
    2\frac{\partial}{\partial \sigma} \lrslp & = 
    \frac{\partial}{\partial \beta} \lrslp + 
    \frac{\partial}{\partial \alpha} \lrslp 
    = 2 - \cos\gamma_{\tL} - \cos\gamma_{\tR} \geq 0 \\
    \label{eq:rsl-ddelta}
    2\frac{\partial}{\partial \delta} \lrslp & = 
    \frac{\partial}{\partial \beta} \lrslp - 
    \frac{\partial}{\partial \alpha} \lrslp 
    = - \cos\gamma_{\tL} + \cos\gamma_{\tR}.
  \end{align}
  Inequalities~\eqref{eq:dlr-dsigma} and~\eqref{eq:rsl-dsigma} imply
  that $\lambda \geq 0$, so let us consider~\eqref{eq:dlr-ddelta}.  By
  Lemma~\ref{lem:case-b-alpha-monotone}, we have $\sigma + \delta =
  \blr(\alpha_0) \geq \blr(\as) = 2\pi -\as > 3\pi/2$ for $d < 2$, and
  so $\cos\delta < \cos(3\pi/2 - \sigma) = -\sin\sigma$.  It
  follows that $2\cos\delta < -2\sin\sigma < -d\sin\sigma$, and
  so~\eqref{eq:dlr-ddelta} is positive.

  We have $0 \leq \gamma_{\tR} \leq \pi$ and we assumed that
  $\gamma_{\tR} + \gamma_{\tL} > 2\pi$.  It follows that $\gamma_\tL >
  2\pi - \gamma_\tR \geq \pi$.  Since $\gamma_{\tR} > 2\pi -
  \gamma_{\tL}$ we have $\cos\gamma_{\tR} < \cos\gamma_{\tL}$. This
  implies that~\eqref{eq:rsl-ddelta} is negative.  But this means that
  $\lambda < 0$, a contradiction.
\end{proof}


\end{document}